%% file: main-arxiv.tex
\documentclass[11pt]{article}
\usepackage[letterpaper,margin=1in]{geometry}
\usepackage[numbers,sectionbib]{natbib}
\usepackage{times}
\usepackage{hyperref}
\usepackage{amsmath,amssymb,amsthm} 
\hypersetup{
pdfborder = 0 0 0,
colorlinks=true,
citecolor=blue,
}
\usepackage{graphicx}
\usepackage{wrapfig,lipsum,booktabs}

\newtheorem{theorem}{Theorem}[section]
\newtheorem{lemma}[theorem]{Lemma}
\newtheorem{corollary}[theorem]{Corollary}
\newtheorem{definition}[theorem]{Definition}

\newenvironment{hypothesis}[1][]{\par\medskip
   \noindent \textbf{Hypothesis.} \em #1 \rmfamily}{\medskip}

 \newenvironment{conjecture}[1][]{\par\medskip
   \noindent \textbf{Conjecture.} \em #1 \rmfamily}{\medskip}

 \newenvironment{result}[1][]{\par\medskip
   \noindent \textbf{Key Result.} \em #1 \rmfamily}{\medskip}

\newcommand{\polylog}{\textrm{polylog}}

\begin{document} 

  \title{On Bioelectric Algorithms:\\
  A Novel Application of Theoretical Computer Science to\\
   Core Problems in Developmental Biology}
  \author{
  Seth Gilbert\\ National University of Singapore\\ {\tt \small seth.gilbert@comp.nus.edu.sg}
   \and James Maguire\\ Georgetown University\\ {\tt \small jrm346@georgetown.edu} 
   \and Calvin Newport\\ Georgetown University\\ {\tt \small cnewport@cs.georgetown.edu}
   }
  \date{}

  \maketitle

\input{abs}

\input{intro}

\input{model}

\input{sym}

\input{type}
\input{tm}
\section{Acknowledgments}
This work was support in part by NSF award \#1649484.

%
\bibliographystyle{plain}
\bibliography{bio,bib}
\end{document}

%% file: abs.tex
\begin{abstract}
Cellular bioelectricity describes the biological phenomenon in which cells in living tissue generate and maintain patterns of voltage gradients induced by differing concentrations 
of charged ions. A growing body of research suggests that bioelectric patterns represent an ancient system that plays a key role in guiding many important developmental processes
including tissue regeneration,  tumor suppression,  and embryogenesis. Understanding the relationship between high-level bioelectric patterns and low-level biochemical processes might
also enable powerful new forms of synthetic biology.
A key open question in this area is understanding {\em how} a collection of cells, interacting with each other and the extracellular environment only through
simple ligand bindings and ion fluxes, can compute non-trivial patterns and perform non-trivial information processing tasks.
The standard approach to this question is to model a given bioelectrical network as a system of differential equations and then explore its behavior using simulation techniques.
In this paper, we propose applying a computational approach. 

In more detail, we present the cellular bioelectric model (CBM), a new computational model that captures
the primary capabilities and constraints of bioelectric interactions between cells and their environment. We use this model to investigate several important topics
in cellular bioelectricity. We begin by studying the ability of an undifferentiated collection of cells to efficiently break symmetry. We prove that a simple
cell definition we call {\tt KnockBack} is remarkably effective at this task. When executed in a single hop topology (all cells can influence each other), {\tt KnockBack} elects a leader in time comparable
to the best solutions in standard computational models. When executed in a multihop topology, it efficiently stabilizes to a {maximal independent set}, even if the cells are started at arbitrary
initial ion concentrations. This latter result is important as these structures have been previously shown to play a role in the nervous system development in flies. 
We then turn our attention to the information processing ability of bioelectric cells.
We provide cell definitions that approximate solutions to the threshold detection and majority detection problems,
and prove that probabilistic solutions with non-zero error are required for these types of problems in this model.
Significantly, we then prove that when it comes to the task of 
computing a function on an input encoded into the cells's initial states,
the CBM is Turing complete.
This result helps resolve an open question about the computational power of simple bioelectric interactions.
\end{abstract}

%% file: intro.tex
\section{Introduction \& Related Work}
\label{sec:intro}

The planarian is an unassuming looking flatworm usually less than an inch long. These simple organisms, however, possess a remarkable ability to regenerate. If, for example, you cut off the head and tail of an unlucky planarian, it will reliably regrow both. 
The biochemical processes that drives this regeneration
start with gene regulatory networks that produce effector proteins, which then inhibit
and promote
intricate endogenous reactions.
Modern biology understand some of these individual steps,
but the process as a whole---with all of its whirling, non-linear,
and unpredictable feedback loops and influences---remains massively too complex for scientists
to usefully decode.

And yet, in a series of extraordinary experiments, systems biologists at Tufts University discovered how to
hijack the planarian regeneration process with simple {\em in vivo} interventions~\cite{durant2017long,biocode}. They can
cut off the head and tail of one of these worms and then induce it to grow two heads, or two
tails, resulting in a perfectly viable, albeit strange new worm.
They have, in some sense, learned how to direct cellular development in these worms without having to master its underlying
biochemical intricacies. 

The secret to these experiments is {\em cellular bioelectricity}: the patterns of voltage differentials caused by
differing concentrations of charged ions inside and outside of a cell's plasma membrane.
A compelling new field of cellular biology, influenced by insights from computer science, 
is revealing that these bioelectric patterns can in some cases play the role of high-level programming languages,
providing a ``biocode" that can specify goal states for cellular development that are then implemented by
complex lower-level processes (see~\cite{levin2017bioelectric} for a recent survey of this work).

In this paradigm, altering the bioelectric pattern (which can be done using interventions
such as chemical blockers that modify ion flux) is like altering the source code of a computer program,
providing a tractable mechanism for controlling how an organism develops.
The ability to manipulate these processes at this high level of abstraction
enables potentially massive breakthroughs in many different important areas of study,
including organ and limb regeneration, tumor suppression, and powerful new forms of synthetic biology.

Some of the key open questions from this research direction 
include understanding {how} cells {\em form} distinct bioelectric patterns, 
and how they {\em alter} them in response to specific environmental inputs.
As detailed below,
the standard approach to exploring these questions applies techniques from dynamical systems theory
(namely simulating systems of differential equations).
In this paper, we instead adopt a biological algorithms approach~\cite{navlakha:2014,navlakha2011algorithms} in which we describe
the systems as distributed algorithms and analyze them with tools from theoretical computer science.
To validate the potential of this strategy we provide both new insight and testable hypotheses for
important open questions from the existing cellular biology literature.

\paragraph{Bioelectric Networks as Dynamical Systems.}
A bioelectric network (BEN) describes a collection of cells along with the parameters
and mechanisms relevant to their bioelectric activity.
These networks typically include {\em ion channels}, which passively enable charged ions to flow between
cells and the extracellular environment, {\em ion pumps}, which actively pump ions against
the gradient induced across the cellular membrane, {\em gap junctions}, which provide direct
connections between cells, and {\em ligands}, which are special molecules that a cell
can release to induce changes in the bioelectric behavior of nearby cells.
Also included in a BEN description are initial concentrations of the relevant
ions (e.g., potassium, sodium, and chorine).

It is not obvious how a particular BEN will behave once its ions, pumps and ligands are allowed to interact and flow.
It is also not obvious how to design a BEN to achieve a particular goal.
With this in mind, the standard way to study these networks is with a dynamical systems approach.
The collection of parameters and mechanisms that make up a particular BEN
are described by a series of complex differential equations. 
Because it is too difficult to calculate analytical solutions to these equations,
the system is studied in simulation (e.g.,~\cite{pietak2017bioelectric,pietak2016exploring}). 
By studying many different configuration and parameters,
the researchers can gain some insights into how a particular type of BEN behaves.

\paragraph{Bioelectric Networks as Distributed Systems.}
An alternative to studying BENs as dynamical systems is the so-called {\em biological algorithms}
approach~\cite{navlakha:2014,navlakha2011algorithms}, 
which describes biological systems as collections of interacting algorithms instead of differential equations.
This approach allows researchers to apply well-established tools from theoretical computer science (and in particular, 
from distributed algorithm theory) to prove strong results about a system's behavior, identify system
designs that can solve specified problems, produce lower bounds and impossibility results,
and even assess the general computational power of the setting in question.
If the computational model that constrains the algorithm designs and dictates their interaction appropriately
abstracts the key features of the biological system it describes, these results can provide useful biological insight and generate testable hypotheses.

In this paper, we are, to the best of our knowledge, the first
to apply the biological algorithms approach to the study of bioelectric networks.
To do so, we begin in Section~\ref{sec:model}
by describing the {\em cellular bioloelectric model} (CBM),
a new computational model
that abstracts the important capabilities and constraints of cellular bioelectrical networks. 
This model assumes a collection of {\em cells} which are connected in a network topology
that describes which cell pairs can directly interact (e.g., through ligand signaling).
To simplify the model specification, time proceeds in synchronous rounds.
The state of each cell at the beginning of a round is captured by a single value that describes
the voltage {\em potential} across its plasma membrane.
A {\em gradient} parameter captures the rate at which this potential increases or decreases
toward an equilibrium in each round due to ion flux through ion channels in its membrane.

Cells can communicate and compute only through {\em bioelectric events},
in which a cell can induce a sudden increase or decrease to its potential
(e.g., by pumping ions in/out, or opening/closing ion gates),
and release ligand molecules that can induce a sudden potential changes
in its neighboring cells in the network.
For each cell, and each bioelectrical event, 
a probability function specific to that event
maps the cell's current potential to the probability of the event firing.
To maintain biological plausibility,
our model requires that these probability function are monotonic,
and allows each cell definition to include only a constant number of distinct bioelectric events.

Though the core computational process in the CBM---the {\em cell}---is quite simple and restricted,
we are able
to show that they are  well-suited to exactly the types of distributed computational
tasks that researchers now attribute to bioelectric behavior.
Below we summarize our results and emphasize the concrete connections they form to active
areas of biological inquiry.

\paragraph{Our Results: Symmetry Breaking.}
As mentioned, one of the key open problems in cellular bioelectrics is understanding
how otherwise identical cells distinguish themselves into set patterns.
We study these symmetry breaking tasks in Section~\ref{sec:sym},
focusing in particular on the {\tt KnockBack} cell definition (see Section~\ref{sec:sym:kb}).
This definition captures one of the simplest possible symmetry breaking strategies.
Cells start with a low potential that gradually increases toward a higher equilibrium.
As a cell's potential increases, it passes through a {\em competition} value range in which,
with constant probability, it fires a bioelectric event that bumps up its potential and emits a ligand
that will reduce the potential of nearby cells.
If it makes it through the competition range, its potential is now high enough that the
 cell will begin firing this event with probability $1$ until its reaches a threshold after which it can
  begin a morphological transformation into a leader.

Though simple, {\tt KnockBack} turns out to be an effective symmetry breaker.
In Section~\ref{sec:sym:le}, we study this strategy in a single hop (i.e., fully connected) network topology.
We prove that not only does it safely elect a single cell to be leader,
it does so in only $O(\log{(n/\epsilon)})$ rounds, with probability at least $1-\epsilon$,
where $n$ is the network size.
For high probability (i.e., $\epsilon < 1/n$),
this bound is {\em faster} than the $O(\log^2{n})$-round algorithm from our recent study 
of symmetry breaking with constant-size state machines~\cite{gilbert:2015}.
It also matches the 
optimal $\Theta(\log{n})$ bound on leader election with unrestricted state machines
under the comparable network
assumptions of a shared communication channel and collision detection~\cite{newport:2014}.

In Section~\ref{sec:sym:mis},
we turn our attention to the behavior of {\tt KnockBack} in connected multihop networks
that satisfy the natural unit ball graph constraints~\cite{kuhn:2005} (which requires the topology
to be compatible with the embedding of the cells in a reasonable metric space).
In this setting, we consider the {\em maximal independent set} (MIS) problem,
in which: (1) every cell must either become a leader or neighbor a leader; (2) no 
two neighbors are leaders.
Our consideration of the MIS problem is not arbitrary.
A 2011 paper appearing in the journal {\em Science}~\cite{afek:2011}
conjectures that nervous system development in flies solves the MIS problem on a layer of epithelial
cells to evenly spread out sensory bristles, motivating the investigation of biologically plausible
strategies for solving this classical problem (c.f.,~\cite{afek:2013,scott:2013}).

We show, perhaps surprisingly, that the simple {\tt KnockBack} strategy turns out to provide
an effective solution to the MIS problem as well.
In more detail, we prove that with high probability in the network size $n$,
it establishes a valid MIS in at most $O(\text{polylog}(\Delta)\log{n})$ rounds,
where $\Delta$ is the maximum degree in the network (which in many biological settings,
such as in~\cite{afek:2011}, it is likely a small constant).

Equally important for the study of bioelectrics, we show this strategy to be self-stabilizing. Even if you start each cell at an
arbitrary initial potential, the system will efficiently stabilize back to a valid MIS.
To the best of our knowledge, {\tt KnockBack} in the CBM is the first efficient MIS solution for a bio-plausible 
or bio-inspired model that is self-stabilizing\footnote{MIS algorithms in the LOCAL model (e.g., ~\cite{Schneider:2008} are often in fact self-stabilizing, even if they are not always explicitly described as stabilizing.}, and is unique in requiring only a single probability value (as opposed to 
the $\log{n}$ distinct probabilities used in most existing efficient solutions, including those proposed in~\cite{afek:2011,afek:2013,scott:2013}). 
Given these powerful properties of the {\tt KnockBack} strategy,
plus a simplicity in design that makes it an easy target for natural selection to identify, 
we make the following hypothesis:

\begin{hypothesis}
The MIS generation in fly nervous system development is potentially driven by robust bioelectric interactions
instead of the more traditional chemical signaling suggested in previous work.
\end{hypothesis}

\noindent This hypothesis is likely testable by using standard fluorescing techniques that help visualize voltage gradients in living
tissue, though we note adapting these tests to flies might require non-trivial lab innovation.

\paragraph{Our Results: Information Processing.}
The other previously mentioned key open problem in cellular bioelectrics is understanding
the capacity of cells to process information using bioelectric interactions.
One conjecture is that simple interactions of the type captured in the CBM are not capable
of much more than simple pattern generation (e.g., generating an MIS with simple {\tt KnockBack}
cells). A competing conjecture is that these interactions are actually capable of performing a wide variety
of non-trivial computation.

In this paper, we use the CBM to provide support for the latter view of biological reality.
We begin in Section~\ref{sec:type} by studying input type computation,
a simple form of information processing also studied in the biologically-plausible population protocol and chemical reaction network 
models (see model comparison below).
In input type computation, the goal is to compute an output based on the {\em number} of cells in the system
of one or more designated types. Two classical problems of this type are {\em threshold detection}~\cite{AngluinADFP06},
which computes whether the number of {\em sick} cells in the system is beyond a fixed threshold $k$,
and {\em majority detection}~\cite{AngluinAE08}, which computes whether there are more $A$ cells than $B$ cells in the system.

We study threshold detection in Section~\ref{sec:type:threshold}.
For small thresholds, we present a simple cell definition that solves the problem exactly with no error probability.\footnote{In
 this context, ``small" means that $k$ is smaller than the maximum number of different ligand counts
a cell can can distinguish, allowing the cell to directly count the sick cells (see the definition of 
{\em binding bound} from Section~\ref{sec:model}).}
For larger thresholds, we present a cell definition that for any error bound $\epsilon$,
correctly detects the threshold is exceeded if the count $n$ is greater than $k\tau$,
and correctly detects it is not exceeded if $n < k/\tau$, for $\tau = O(\log{(1/\epsilon)})$.
We conclude by proving that {\em any} solution to this problem that works for general $k$ values
must have a non-zero error probability, regardless of how large we allow $\tau$ to grow.

In Section~\ref{sec:type:majority},
we turn our attention to majority detection.
We provide symmetric cell definitions for type $A$ and $B$ cells.
 For any constant error bound $\epsilon >0$,
these cells will correctly detect the majority type with probability $1-\epsilon$
so long as there are a constant factor more of the majority type
(for a constant factor defined relative to $\ln{(1/\epsilon)}$).

The general threshold detection solution is straightforward: cells send a 
ligand with probability $1/k$, and associate any received ligands with an exceeded threshold.
The majority detection solution has cells increase the firing probability of a bioelectric event from a small lower bound
to a constant as their potential increases towards equilibrium: whichever cell type fires first is assume the majority type.
In both cases, more refined probabilistic analysis would likely lead to tighter bounds,
but the solutions and lower bound in Section~\ref{sec:type} are sufficient to generate the following
conjecture:

\begin{conjecture}
Bioelectric cells can approximate many of the standard input type computational problems,
but will require non-zero error probabilities to do so.
\end{conjecture}

In Section~\ref{sec:tm} we consider a more general form of information processing,
in which the input value to be processed in a given execution is encoded in the initial
value of one or more designated input cells (for some encoding scheme specified by 
the designer of the cellular system).
Understanding the set of functions that  can be computed by such systems
provides insights into the computational power of bioelectrics.
With this motivation in mind, we prove the following strong result:

\begin{result}
Bioelectric cells are Turing complete.
\end{result}

In slightly more detail, we prove that for any deterministic Turing machine (TM) $M$,
there exists a finite collection of cells including a designated {\em input cell},
connected in a single hop network,
such that for any TM input $w$,
 if you set the initial potential value of the input cell to a proper encoding of $w$,
 the system will correctly simulate $M$ on $w$.
 Of course, one of the TMs that can be simulated is a universal TM, 
 indicating the existence of a computationally universal collection of bioelectric cells.
 
 Our strategy for this result is to first apply a result due to Minsky~\cite{minsky1967}
 to convert $M$ into an equivalent counter machine in which the TM input is encoded
 into the initial count of a designated counter.
 We describe a novel strategy for simulating the machine using cells in the CBM.
 We dedicate one cell for each simulated counter, where the cell potential corresponds
 to the counter value, one cell for each state of the finite state machine capturing
 the counter machine control logic, and one cell for each machine transition.
The cells interact using ligands in a fine-tuned pattern that correctly computes the next state
and correctly updates the counter values.
We emphasize the simplicity of the cells used in this simulation: no cell definition includes more
than two bioelectric events or requires a cell to react to more than two ligand types; all bioelectric
events are controlled by deterministic step functions.
 
The goal with this result, of course, is not to imply that actual biological systems
are implementing TM simulations of this style (counter machine simulations are inefficient). It is instead meant to help resolve
the open question about the theoretical ability of simple bioelectric interactions to implement
complex computations.
 
 \paragraph{Comparison to Existing Models.}
Generally speaking, in studying the intersection of biology and algorithms there are two main types
of computational models used: those with {\em bio-plausible computation} and those with 
{\em bio-plausible constraints}.
The first category describes models in which 
the actual method of computation is motivated by a specific biological context.
Algorithms in these models cannot simply be described in standard pseudocode
or state machine descriptions. They must instead be specified in terms of the
particular bio-plausible computation method captured by the model.
Well-known models of this type include neural 
networks~\cite{rosenblatt1958perceptron,siegelmann1991turing,maass1997networks,lynch:2017},
 chemical reaction 
 networks~\cite{soloveichik2008computation,chen2014deterministic,chen2017speed}, and population protocols~\cite{AngluinADFP06,AngluinAE08a,AngluinAE08,AngluinAER07,AngluinAE06,ChatzigiannakisS08}
  (which
 are computationally equivalent to certain types of chemical reaction networks).
 
 The other type of model used to study biological algorithms are those with bio-plausible
 constraints. These models describe computation with the same standard discrete state machine
 formalisms assumed in digital computers. They constrain algorithms, however,
 by adding biologically-motivated limits on parameters such as memory size, 
 the message alphabets used for communication, and the behavior of the communication channels.
 Well-known models of this type includes the ANTS  
 model~\cite{feinerman2012collaborative,lenzen2014trade,cornejo2014task,ghaffari2015distributed,shiloni2011robot},
 the stone age computing model~\cite{Emek:2013}, and the beeping 
 model~\cite{degesys:2007,degesys:2008,MotskinRSG09,CornejoK10,AfekABCHK11,ScottJ013,AfekABCHK13,gilbert:2015}.
 
 Both models are useful for applying algorithmic tools to understanding biological systems.
 The bio-plausible computation models focus more on understanding the low level processes behind
 particular behaviors, while the bio-plausible constraints models focus more on identifying general distributed
 strategies, and understanding the minimum
 resources/assumptions required for useful distributed coordination.
 
 The CBM is most accurately categorized as a bio-plausible computation model.
 Existing studies of the stone age and beeps computing models already shed light
 on what can be computed by collections of simple state machines with basic signaling capabilities.
 The goal of our work is to understand what can be computed with the {\em specific} bioelectric
 mechanisms implemented in living tissue. This goal is important as our work is designed
 to be relevant to system biologists that are studying and manipulating these specific mechanisms.

 A key factor differentiating the CBM from existing models is that a cell is not a state machine (e.g., unlike in beeping models, stone age models, ANTS models, and population protocols).  In fact, a cell in our model is {\em computationally incomparable} with a traditional state 
 machine.  Consider a basic task such as outputting a repeated pattern: \texttt{ABCABCABC\ldots}.  This is trivial to implement with a discrete state machine: cycle through three states,
 one dedicated to each symbol. It is not hard to show, however, that this behavior cannot be implemented by
 a cell in the CBM.  The key difficulty is the required monotonicity for firing functions driving bioelectric events (which is an important property of the real biological cells being modelled). A simple argument establishes that for any cell, there must be two symbols $S_1,S_2 \in \{A,B,C\}$,
such that whenever $S_1$ has a non-zero probability of being output, so does $S_2$, and so you will not get a perfectly repeating pattern.  At the same time, we cannot necessarily simulate a cell with a finite state automaton, since each cell stores an analog potential value.   It is therefore unclear how to use an existing bio-plausible constraint model to directly explore bioelectric dynamics.
 
 Finally, we note there are interesting connections between the CBM and neural network models.
 The action potential that drives neural computation is itself a bioelectric behavior.
 Indeed, many of common neural network models can be understood
 as special cases of our model in which the ``algorithm" designer gets to specify the network structure, 
 and the probabilistic firing functions for the bioelectric event describing the action potential must be of
 a specified form (e.g., a step or sigmoid).
 It is possible, therefore, that the CBM might have a role to play in better understanding neural computation,
 Un this paper, we leverage its more general definition of a CBM to study bioelectric behavior in 
the non-neural context.

%% file: model.tex
\section{The Cellular Bioelectric Model}
\label{sec:model}

Here we define the cellular bioelectric model (CBM),
a synchronous computation model that abstracts the key capabilities and constraints
of bioelectric networks (see Section~\ref{sec:intro} for a more detailed discussion and motivation of these networks and their behavior).
We begin by defining the main computational unit of this the model, the cell, and then describe executions of systems of cells connected by a network topology.
We conclude by defining some additional biology-inspired constraints and  capabilities.

\paragraph{Cells.}
Fix a non-empty and finite set $L$ containing the {\em ligands} cells use to drive bioelectrical interactions.
We define a {\em bioelectric event} to be a pair $(f,(\delta,s))$,
where $f:{\mathbb R} \rightarrow [0,1]$ is a {\em firing function} from real numbers to probabilities,
and $(\delta, s)$ consists of a {\em potential offset} value $\delta \in {\mathbb R}$,
and a {\em ligand} $s\in L$.
We also define a {\em membrane function} to be a function $g$
from multisets defined over $L$ to real numbers.

Pulling together these pieces, a {\em cell} in our model is described by a $6$-tuple $(q_0, \sigma, \lambda, \omega, g, {\cal B})$,
where $q_0\in {\mathbb R}$ is the {\em initial potential} value of the cell,
$\sigma \in {\mathbb R}$ is the {\em equilibrium} potential that the cell will drive its internal potential toward (i.e., through
ion flux),
$\lambda\in {\mathbb R}^{+}$ is a non-negative real number describing
the {\em gradient} rate at which the cell's potential moves toward $\sigma$,
$\omega \in {\mathbb R}$ is the smallest possible potential for the cell,
$g$ is a membrane function, and ${\cal B}$ is a set of bioelectric events. 
For a given cell c, we use the notation $c.q_0, q.\sigma, q.\lambda, c.\omega, c.g, c.{\cal B}$ to 
refer to these six elements of the cell's tuple.

\paragraph{Systems and Executions.}
A {\em system} in our model consists of a non-empty set ${\cal C}$ of $n=|{\cal C}|$ cells,
 an undirected graph $G=(V,E)$ with $|V| = n$,
 and a bijection $i: {\cal C} \rightarrow V$ from cells to graph vertices.
 For simplicity,
 in the following we sometimes use the terms {\em cell $u$} or {\em node $u$}, for some $u\in V$,
 to refer to the unique cell $c\in {\cal C}$ such that $i(c) = u$.

An {\em execution} in our model proceeds in synchronous rounds that we label $1,2,3,...$.
At the beginning of each round $r$,
we define the {\em configuration} $C_r: {\cal C} \rightarrow {\mathbb R}$ as the bijection from cells to their potential
values  at the beginning of round $r$.
For each $c\in {\cal C}$, $C_1(c) = c.q_0$. That is, each cell starts with the initial potential value provided as part of its definition.
The configuration for each round $r>1$ will depend on the configuration
at the start of round $r-1$, and the (potentially probabilistic) behavior of the cells during round $r-1$.

\bigskip

\noindent In more detail, each round $r\geq 1$ proceeds as follows:

\begin{enumerate}

	\item For each cell $c\in {\cal C}$, initialize $p_c \gets C_r(c)$ to $c$'s potential at the start of round $r$.
	We will use $p_c$ to track how $c$'s potential value changes during this round.
	Also initialize multiset $M_c=\emptyset$. We will use $M$ to collect ligands sent toward $c$  during this round.
	
	\item For each cell $c\in {\cal C}$,
	and each bioelectric event $(f,(\delta,s)) \in c.{\cal B}$,
		this event {\em fires} with probability $f(C_r(c))$.
		If the event fires, update $p_c \gets p_c + \delta$ and add a copy of $s$ to multiset $M_{c'}$,
		for each cell $c'\in {\cal C}$ such that $\{ i(c), i(c')\} \in E$ 
		(that is, for each cell $c'$ that neighbors $c$ in $G$).
		
	\item After processing all rules at all cells, the round proceeds by having
	cells process their incoming ligands.
	For each cell $c\in {\cal C}$,  update $p_c \gets p_c + c.g(M_c)$.
	That is, update the potential change according to $c$'s membrane function applied to its incoming ligands.
	
	\item Finally, we calculate the impact of the gradient driving each cell $c$'s potential toward
	its equilibrium value. 
	In more detail, let $z = C_r(c) - c.\sigma$. We define the gradient-driven potential change
	for $c$ in round $r$, denoted $\lambda_r(c)$, as follows:
	
	 \[ \lambda_r(c)  \gets \begin{cases}
	 	-c.\lambda & \text{if $z \geq c.\lambda$} \\
		-z & \text{if $0 < z < c.\lambda$} \\
		0 & \text{if $z = 0$} \\
		z & \text{if $-c.\lambda < z < 0$} \\
		c.\lambda & \text{if $z \leq -c.\lambda$}
	 \end{cases} \]
	
	We add this gradient-induced offset to $c$'s potential: ,$p_c \gets p_c + \lambda_r(c)$.
	
	\item The final step is to the initial potential for $r+1$ for each $c\in {\cal C}$,
	by performing a final check that
	the potential did not fall below the cell's lower bound in the round:
	  $C_{r+1}(c) \gets \max\{p_c, c.\omega\}$.

\end{enumerate}

\paragraph{Natural Constraints on Cell Definitions.}
To maintain biological plausibility, our model includes the following
natural constraints on allowable cell definitions:

\begin{itemize}

	\item {\em Constraint \#1:} Each cell definition includes at most a constant number
	of bioelectric events. 
	\item {\em Constraint \#2:} Firing functions are monotonic.
	\item {\em Constraint \#3:} For each membrane function $g$,
	there must exist some constant $b>0$, 
	such that for every possible ligand multiset $M$, $g(M) = g(\hat M)$,
	where $\hat M$ is the same as $M$ except every value that appears {\em more}
	than $b$ times in $M$ is replaced by
	exactly $b$ copies of the value in $\hat M$. 
	We call the value $b$ the {\em binding bound} for that cell definition.
	
	
\end{itemize}

\paragraph{Expression Events \& Thresholds.}
In real biological systems,
bioelectric patterns induce morphological changes driven by lower-level processes.
To capture this transformation we introduce the notion of {\em expression events}
into our model (named for the idea that bioelectics regulates gene {\em expression}).

In more detail, some of our problem definitions specify a potential threshold such
that if a cell's potential exceeds this threshold, an irreversible morphological transformations begins.
This check occurs at the beginning of each round. That is, if a cell begins round $r$ with a potential
value that exceeds the event threshold, we apply the event.
For example, in studying leader election (see Section~\ref{sec:sym}),
we assume once a cell passes a given threshold value with its potential it transforms into a {\em leader},
at which point it stops executing its original definition and transforms neighbors that have potential values
below the threshold into {\em non-leaders}.
The specification and motivation for specific expression thresholds are included as part of the problem definitions
for the problems studied in this model.

%% file: sym.tex
\section{Symmetry Breaking}
\label{sec:sym}

A fundamental task in bioelectric networks is generating non-trivial bioelectric
patterns that can then direct  cellular development. 
This requires symmetry breaking.
With this in mind,
as detailed and motivated in Section~\ref{sec:intro},
we study here the symmetry breaking capabilities of a simple but surprisingly effective
cell definition called {\tt KnockBack}.
We define this cell in Section~\ref{sec:sym:kb},
then study its ability to elect a leader in single hop networks in Section~\ref{sec:sym:le},
and study  its ability to efficiently generate maximal independent sets in multihop networks in Section~\ref{sec:sym:mis}.


\subsection{The {\tt KnockBack} Cell Definition}
\label{sec:sym:kb}

We define a {\tt KnockBack} cell as follows:

\begin{center}
\begin{tabular}{|cc|}
\hline
\multicolumn{2}{|c|}{{\tt KnockBack} cell definition} \\
\hline
\hline
$q_0 = 0$ & ${\cal B} = \{(f, (1/2, m))\}$, where: \\
$\lambda = 1/2$, $\sigma=2$, $\omega=-2$ & \multicolumn{1}{r|}{$f(x < 1/2) = 0$}\\
%
$g(|M| > 0) = -(3/2)$ & \multicolumn{1}{r|}{$f(1/2 \leq x < 1) = 1/2$}\\
$g(|M| = 0) = 0$ & \multicolumn{1}{r|}{$f(x \geq 1) = 1$}\\
\hline
\multicolumn{2}{|c|}{leader expression rule threshold: $\geq 2$}\\
\hline
\end{tabular}
\end{center}

The {\tt KnockBack} cell implements an obvious symmetry breaking strategy.
It is initialized with a low initial value of $q_0 = 0$ that is driven toward the equilibrium of $\sigma = 2$
at a gradient rate of $\lambda = 1/2$.
As a cell's potential value passes through the range of $[1/2,1)$,
its single bioelectric event $(f, (1/2, m))$ fires with constant probability.
If this event fires, the cell increases its potential by $1/2$ (e.g., by pumping in more ions), 
and emits the ligand $m$, which will bind with its neighbors in the network.
If at least one of the cell's neighbor emits the ligand $m$, then that cell will decrease its potential by $-(3/2)$ (e.g., by pumping out ions).

If a cell makes it to a potential value of $1$ or greater,
this event starts firing with probability $1$.
If a cell makes it to potential value of  $2$ or greater,
it executes the {\em leader expression event}, which makes it a leader,
and makes each neighbor below the threshold into non-leaders.

As will be elaborated in the analyses that follow,
this cell definition prevents neighbors from both becoming leaders because any cell that becomes a leader in some round $r+1$,
must have spent round $r$ at a potential value where it fires its bioelectric event with probability $1$.
If two neighbors fire this event during $r$, however, they both would have ended up with a net decrease in their potential, preventing
them from becoming leaders in $r+1$.

The time required for a leader to emerge is more complicated to derive, especially in the multihop context.
The intuition behind these analyses, however, is that when multiple nearby cells simultaneously have potential
values in the {\em competition range} of $[1/2, 1)$,
it is likely that some will fire their event and some will not.
Those that do fire their event end up with a smaller net decrease in their potential than those that did not---reducing
the set of cells at the front of the pack competing to make it past the competition range.

\input{le}

\input{mis}


%% file: le.tex
\subsection{Single Hop Leader Election}
\label{sec:sym:le}

Consider a single hop (i.e., fully-connected) network  consisting of $n>0$ copies of the 
{\tt KnockBack} cell defined in Section~\ref{sec:sym:kb}.
We study the ability of this system to solve the leader election problem,
which requires the system to converge to a state in which one cell is a leader and all other cells are non-leaders.
We prove that the system never elects more than one leader,
and the for any error probability $\epsilon >0$, 
with probability at least $1-\epsilon$ it elects a leader in $O(\log{(n/\epsilon)})$ rounds.
As we detail in Section~\ref{sec:intro},
this round complexity is comparable to the best-known solutions in more powerful computational models.

\paragraph{Preliminaries.}
We begin our analysis with  some useful notation and assumptions.
First, if a cell $c$ passes the threshold of $2$ needed to trigger the leader expression event, we say it is {\em elected leader}.
Any cell $c'$ that has a potential below $2$ at the start of a round in which some other cell $c$ is elected leader,
becomes a {\em non-leader}. We assume once a cell becomes a leader or non-leader it stops executing
{\tt KnockBack} and becomes quiescent.
We say the system is {\em active} in round $r$ if no cell has been elected a leader in the first $r-1$ rounds,
and no cell start starts round $r$ with a potential large enough to for it to be elected leader.

We next note that the amount a cell $c$'s potential changes in some active round $r$ is entirely
determined by two factors: whether or not $c$ sends a ligand, and whether or $u$ receives at least one ligand.
The following table captures the net change to $c$'s potential in a given active round based on the combination
of these two factors. Notice that included in this net potential is $c$'s positive gradient increase of $1/2$,
the increase of $1/2$ that happens if it sends, and the decrease of $-3/2$ that happens if it receives:

\bigskip

\begin{center}
\begin{tabular}{|c|c|c|}
\hline
                    & {\bf send} & {\bf no send} \\
\hline
{\bf receive} &  $-1/2$ & $-1$ \\
\hline
{\bf no receive} & $1$ & $1/2$ \\
\hline
\end{tabular}
\end{center}

\bigskip

\noindent We call the above the {\em net potential table.} Differentiating the cases captured in this table will be useful in the analysis that follows.

We now introduce some useful terminology for tracking the maximum potential values in the system from round to round.
For the sake of completeness in the following definitions, we adopt the convention that if a cell $c$ becomes a leader or non-leader in round $r$,
then its potential value freezes at $C_r(c)$ for the remainder of the execution; i.e., $\forall r' \geq r: C_{r'}(c) = C_r(c)$. 

\begin{definition}
For every round $r \geq 1$, we define:
\begin{itemize}
	\item $\hat p(r) = \max_{c\in {\cal C}}\{C_r(c)\}$
	\item $A(r) = \{ c\in {\cal C}: C_r(c) = \hat p(r)\}$
	\item $B(r) = \{c' \in {\cal C}:C_r(c') < \hat p(r)\}$
	\item $send(r)= \{c\in {\cal C}: \text{$c$ sends a ligand in round $r$} \}$
	
\end{itemize}
\end{definition}

Finally, we identify the special case of a round in which cells decide to send probabilistically.
These are the rounds that can advance a system toward a leader.

\begin{definition}
We define an round $r$ to be a {\em competition round}, if and only if $\hat p(r) = 1/2$.
\end{definition}

\paragraph{Proof of Safety.}
We now prove that at most one cell becomes leader.

\begin{theorem}
A single hop network comprised of {\tt KnockBack} cells never elects more than one leader.
\label{thm:le:safety}
\end{theorem}
\begin{proof}
Assume for contradiction that two different cells $c$ and $c'$ both become a leader during the same active round $r>1$.
All cells start with potential $0$, and, as captured in the net potential table, can increase or decrease this potential only by multiples of $1/2$.
Furthermore, the largest possible increase in a given round is $1$.
We assumed that round $r$ was the first round that $c$ and $c'$ start with a potential of at least $2$.
It follows that $c$ and $c'$ must have each started round $r-1$ with a potential in $\{1,3/2\}$.
According to the net potential table, however, it follows that $c$ and $c'$ would have both {\em decreased} their potential 
by $1/2$ during round $r-1$, starting round $r$ with a potential in $\{1/2, 1\}$, contradicting the assumption that both 
cells are elected leader in $r$.
\end{proof}

\paragraph{Proof of Liveness.}
We now show that it does not take too long to elect a leader with reasonable probability.
In more detail, our goal is to prove the following:

\begin{theorem}
Fix some error bound $\epsilon > 0$ and network of $n\geq 1$ {\tt KnockBack} cells.
With probability at least $1-\epsilon$,
a leader is elected within $O(\log{(n/\epsilon)})$ rounds.
\label{thm:le:liveness}
\end{theorem}

To prove this theorem, we leverage the definitions from the above preliminaries section, 
to establish that $B$ is an absorbing state.

\begin{lemma}
Fix some active round $r\geq $ and cell $c\in B(r)$.
For every $r' > r$: $c\in B(r')$.
\label{lem:le:absorb}
\end{lemma}
\begin{proof}
Fix some $r$, $c$ and $B(r)$ as specified above.
To prove the lemma it is sufficient to show that $c\in B(r) \Rightarrow c\in B(r+1)$,
as the same argument can then be reapplied inductively to achieve the property for any $r' > r$.

Because $c\in B(r)$ it follows that $C_r(c) < \hat p(r)$,
which further implies that $A(r)$ is non-empty.
We proceed with a case analysis on $\hat p(r)$.
In the following, we leverage the property that all potential values are multiples of $1/2$.

\begin{itemize}
	\item If $\hat p(r) < 1/2$, then no cell sends a signal in $r$, and, according to the net potential table, all cells increase their potential
	by $1/2$, preserving the gap between $C_r(c)$ and $\hat p(r)$.
	
	\item If $\hat p(r) = 1/2$, then $C_r(c) \leq 0$. It follows that cell $c$ does not signal (by emitting a ligand) during round $r$. There are two relevant sub-cases
	based on the behavior of the cells in $A(r)$. If no cell in $A(r)$ signals, then all cells once again increase their potential by $1/2$,
	preserving the gap between $C_r(c)$ and $\hat p(r)$. If at least one cell in $A(r)$ signals,
	then $u$ decreases its potential by $1$ while the sender(s) decrease their potential by at most $1/2$ (and, in the case of a single sender,
	actually increase potential by $1$).  Either way the gap between $C_r(c)$ and $\hat p(r)$ grows in this round.
	
	\item If $\hat p(r) \geq 1$, then all cells in $A(r)$ signal. If $u$ also sends, then all cells reduce their potential by $1/2$, 
	preserving the gap between $C_r(C)$ and $\hat p(r)$.
	If $u$ does not send, then as argued in the above case, this gap grows.
		
\end{itemize}

In all cases, $c\in B(r+1)$. 
\end{proof}

Having established the dynamics between the $A$ and $B$ partitions,
we now focus on the competitions that move cells between these sets,
proving that competition rounds happen frequently until a leader is elected.

\begin{lemma}
Fix some round $r\geq 1$.
If $r$ is a competition round then either: $r+2$ is a competition round or a leader is elected by $r+2$.
\label{lem:le:comp}
\end{lemma}
\begin{proof}
Consider the possibilities for the signaling behavior in competition round $r$.
If $|send(r)| = 0$, then by the net potential table,
all cells increase their potential by $1/2$.
It follows that the cells in $A(r+1) = A(r)$ now have potential value $1$,
while the cells in $B(r+1)$ have potential values no larger than $1/2$.
All cells in $A(r+1)$ will therefore send a ligand. It is possible that some cells in $B(r+1)$ send as well.
There are two relevant sub-cases for round $r=1$. 

\begin{itemize}
\item If $|send(r+1)| > 1$, then all cells decrease their potential by $1/2$.
It follows that the cells in $A(r+2) = A(r+1) = A(r)$ decrease their potential to $1/2$.
Therefore, round $r+2$ is once again a competition round.

\item On the other hand, if $|send(r+1)| = 1$,
the single sender $c\in A(r+1)$ will increase its potential by $1$ and therefore be elected leader in round $r+2$.

\end{itemize}

Another possibility for competition round $r$ is that $|send(r)| = 1$.
In this case, let $c\in A(r)$ be the single cell that sends a ligand in $r$.
By the net potential table, it increases its potential by $1$ while all other cells decrease their potential by $1$.
It follows that $C_{r+1}(c) = 3/2$, while for all $c' \neq c, C_{r+1}(c') \leq -1/2$.
Cell $c$ will sends alone in round $r+1$, pushing its potential beyond $2$, and therefore
electing it leader in round $r+2$.

The final possibility is that $|send(r)| > 1$.
In this case, by the net potential table, all cells decrease their potential by $-1/2$.
It follows that $\hat p(r+1) = 0$. During round $r+1$, therefore, no cell signals.
Therefore, all cells increase their potential by $1/2$.
Accordingly, $\hat p(r+2) = 1/2$, making $r+1$ a competition round.
\end{proof}

Another important property of competition rounds is that with constant probability they reduce
the number of cells in $A$ by a constant fraction due to the case in which some cells send a ligand and some do not.

\begin{lemma}
Fix some competition round $r \geq 1$.
Let $n_r = |A(r)|$ be number of cells that might potentially send a ligand in $r$.
If $n_r > 1$, then
with probability greater than $1/12$:  $0 < |send(r)| < (3/4)n_r$.
\label{lem:le:prob}
\end{lemma}
\begin{proof}
Fix some competition round $r$ and $n_r = |A(r)|$ as specified in the lemma statement.
In the following we call the cells in $A(r)$ the competitors for the round.
Let $X_r = |send(r)|$ be the random variable that describes the number of competitors that send ligands in round $r$.
 Because each competitor sends with probability $1/2$,
we know $E(X_r) = n_r/2$.
To establish an upper bound on this count,
we apply Markov's inequality:

\[  \Pr\left(X_r \geq (3/2)E(X_r)\right) \leq \frac{E(x)}{(3/2)E(x)} = 2/3.\]

\noindent Stated another way: the probability that $X_r \geq (3/2)E(X_r) = (3/4)n_r$ is no more than $2/3$.
From the lower bound perspective, we  leverage the assumption that $n_r \geq 2$ to establish the following:

\[ \Pr\left( X_r = 0 \right) = 2^{-n_r} \leq 1/4.  \]

\noindent Finally, applying a union bound to combine these bounds, it follows:

\[ \Pr\left( X_r = 0 \vee X_r \geq (3/4)n_r \right) < 2/3 + 1/4 = 11/12. \]

\noindent Therefore, we satisfy the constraints of the lemma statement with a probability greater than $1/12$, as required.
(Notice that this bound is relatively loose. Replacing a union bound with
a stricter concentration argument would yield a larger constant.
But for the purposes of our analysis, which focuses on the asymptotic dynamics of this system,
any reasonable constant here is sufficient.)
\end{proof}

We now have all the pieces needed to prove Theorem~\ref{thm:le:liveness}.

\begin{proof}[Proof (of Theorem~\ref{thm:le:liveness})]
By the definition of {\tt KnockBack},
round $2$ is a competition round.
By Lemma~\ref{lem:le:comp}, this is also true of rounds $4$, $6$, $8$, and so on, until a leader is elected.

Fix some such competition round $r$.
Notice, if a cell $c\in A(r)$ does not send a ligand in $r$,
but $|send(r)| > 0$,
then $\hat p(r+1) - C_{r+1}(c) \geq 1/2$,
meaning that $c \in B(r+1)$.
By Lemma~\ref{lem:le:absorb},
$c$ will never again be in set $A$.

In this case, we say $c$ is {\em knocked out} in competition round $r$.
More generally, Lemma~\ref{lem:le:absorb} implies that the $A$ sets are monotonically non-increasing.
That is, $A(1) \supseteq A(2) \supseteq A(3) \supseteq ...$

Assume we arrive at some competition round $r$ such that $A(r) = \{c\}$, that is, there is a single competitor $c$ left in the system.
It is straightforward to verify that if $|send(r)| = 1$, then $c$ will become leader in round $r+2$:
in round $r$, $c$ increases its potential to $3/2$ while all other cells decrease their potential by $1$ to values lower than $1/2$;
$c$ then sends a ligand alone with probability $1$ during round $r+1$ and becomes leader in $r+2$.
The probability that $c$ sends in $r$ is exactly $1/2$.
Therefore, once we we get to a {\em single competitor} state, we have probability $1/2$ of electing a leader (within two rounds) at each
subsequent competition round.
The probability that we experience $t' = \log{\frac{2}{\epsilon}}$ competition rounds in a single competitor state {\em without} electing a leader is upper bounded
by $(1/2)^{t'} = \epsilon/2$.

We now bound the number of competition rounds required to drive us to a single competitor state with sufficiently high probability. 
If $n=1$, this occurs at the beginning of the first round with probability $1$.
So we continue by considering the case where $n>1$.

To do so, we say a competition round $r$ in this context is {\em productive} if at least $(1/4)$ of the remaining competitors are knocked
out in this round. 
Because $A(1) = n$, after $t$ productive competition rounds, at most $n(1-1/4)^t$ cells remain competitors.
Because $n(1-1/4)^t < n\cdot \exp(-t/4)$, it follows that the total number of productive competition rounds
before arriving at a single competitor state is upper bounded by $t = 4\ln{n}$.
(A key property of any knock out style algorithm is that you can never have a round in which all cells are knocked out,
as to be knocked requires at least one cell that sends and is therefore not knocked out in that round.)

By Lemma~\ref{lem:le:prob},
a given competition round is productive with probability greater than $1/12$.
This lower bound holds regardless of the execution history.

Let random variable $Y_k=\sum_{i}^k X_i$, where each $X_i$ is the trivial independent random indicator variable that evaluates to $1$ with probability $1/12$,
 and otherwise $0$.
Let $p_k$ be the probability that the number of competitors reduces to $1$ after no more than $k$ competition rounds.
A standard stochastic dominance argument establishes that $p_k > \Pr(Y_k \geq 4\ln{n})$.

Consider $k=96\ln{(n/\epsilon)}$.
For this value, $E(Y_k) = 8\ln{(n/\epsilon)}$.
Applying a Chernoff bound,
it follows:

\[  \Pr\left(Y_k \leq E(Y_k)/2 = 8\ln{(n/\epsilon)}\right) \leq \exp\left(  - E(Y_k)/8 \right) = \exp\left( -\ln{(n/\epsilon)} \right) = \frac{\epsilon}{n}.\]

\noindent Because $8\ln{(n/\epsilon)} < 4\ln{n}$ and $\frac{\epsilon}{n} \leq \epsilon/2$,
it follows that for this definition of $k$,
$p_k > \Pr(Y_k \geq 4\ln{n}) > 1- \epsilon/2$.

Combining our two results with a union bound, the probability that $t' + k$ competition rounds is not sufficient to elect a leader
is less than $\epsilon$.
By Lemma~\ref{lem:le:comp},
there is a competition round at least every other round until a leader is elected.
It follows that $O(t'+k) = O(\ln{(n/\epsilon)})$ total rounds is sufficient to elect a leader with probability at least $1 - \epsilon$.
\end{proof}

%% file: mis.tex
\subsection{Maximal Independent Sets}
\label{sec:sym:mis}

We now study the behavior of the {\tt KnockBack} cell 
when executed in a multihop network topology that satisfies a natural constraint defined below.
We show, perhaps surprisingly,
that this simple cell efficiently solves the {\em maximal independent set} (MIS) problem in this context.

In slightly more detail, solving the MIS problem requires that the system satisfy
 the following two properties: (1) {\em maximality}, every cell is a leader or neighbors a leader; and (2) {\em independence},
no two neighbors are leaders.
We prove that the leaders elected by {\tt KnockBack} in a multihop network always satisfy property $2$,
and that with high probability in the network size $n$,
property $1$ is satisfied in $O(\text{polylog}(\Delta)\log{n})$ rounds,
where $\Delta$ is the maximum degree in the network topology (and in many biological contexts, likely a small constant).
We then show that the algorithm still efficiently stabilizes to an MIS even if we start cells at arbitrary potential values.

As we elaborate in Section~\ref{sec:intro},
the simplicity, efficiency, and stabilizing nature of generating MIS's with {\tt KnockBack}
leads us to hypothesize that bioelectrics might play a role in the observed generation of MIS patterns in the epithelial cells of flies~\cite{afek:2011}.
As we also elaborate in Section~\ref{sec:intro},
the round complexity of our solutions is comparable to existing solutions in more powerful computation models.

\paragraph{The Unit Ball Graph Property.}
We study the MIS problem in connected multihop networks that
satisfy the unit ball graph (UBG) property
first described by Kuhn et~al.~\cite{kuhn:2005}.
A graph $G = (V,E)$ is considered a UBG (equivalently, satisfies the UBG property) if it satisfies
the following two constraints:
(1) there exists an embedding of the nodes in $V$ in a metric space
such that there is an edge $\{u,v\}$ in $E$ if and only if $dist(u,v) \leq 1$;
and (2) the doubling dimension of the metric space,
defined as the smallest $\rho$ such that every ball can be covered by at most $2^{\rho}$ balls
of half its radius, is constant.

It is typical to think of cells embedded in two or three-dimensional Euclidean space,
where only nearby cells can directly interact.
Both these spaces satisfy the UBG property.
By assuming a UBG graph, however, 
not only can we produce results that apply to multiple dimensionalities,
we also allow for more general distance function definitions that can model natural obstructions and occasional
longer distance connections.  Our MIS analysis strongly leverage the bounded growth of graphs that satisfy this property.

\paragraph{Preliminaries.}  
Throughout this analysis, we assume that all cells start with initial potential $q_0$ as specified by the {\tt KnockBack} definition.
We will later tackle the case where cells start with arbitrary initial potential values and show that it still stabilizes to an MIS.
In the following, when we say that a cell is \emph{active}, we mean that if has potential $< 2$, and all of its neighbors have
potential $< 2$.  That is, a cell becomes inactive when either it joins the MIS (becomes a leader), or one of its neighbors joins the MIS.

\paragraph{Safety.}

First, we observe that if a cell reaches potential $1.5$, then forever thereafter it continues 
to have high potential, while all of its neighbors remain with negative potential.
\begin{lemma}
\label{lem:knockout}
If cell $c$ in round $r$ increases its potential from a value $< 1.5$ to a value $\geq 1.5$,
then in every round $> r$: (i) cell $c$ maintains a potential $\geq 1.5$, and (ii) every
neighboring cell maintains a potential $< 0$.
\end{lemma}
\begin{proof}
The maximum increase in potential in a round is 1, so cell $c$ must begin round $r$ with potential at least $0.5$.  If cell
$c$ has potential $0.5$, then it can only increase its potential to $1.5$ by broadcasting.  If cell $c$ has potential at least $1$,
then it always broadcasts.  Therefore we can conclude that $c$ broadcasts in round $r$.

If $c$ receives a message in round $r$, then it decreases its potential.  Thus we can conclude that $c$ does not receive a message
in round $r$, i.e., none of its neighbors broadcast.  This implies that every neighbor had potential $\leq 0.5$.  Since each neighbor 
did not broadcast, but received a message from $c$, we know that each neighbor decreased its potential by $1$ and hence has potential $\leq -0.5$.

We can then see, by induction, that this situation contains forever thereafter: in each following round, $c$ broadcasts (since its potential is $\geq 1$) 
and its neighbors do not broadcast (since their potential is $\leq 0$).  Therefore, $c$'s potential remains $\geq 1.5$ and the neighbors potential remains $\leq -0.5$.
\end{proof}

As a corollary, we immediately see that no two neighbors can both be in the MIS:
\begin{lemma}
\label{lem:safety}
Let $c$ and $c'$ be two neighboring cells.  It is never the case $c$ and $c'$ both
have potential $> 1.5$.
\end{lemma}
\begin{proof}
Let $t$ be the first round that either $c$ or $c'$ reaches potential $\geq 1.5$, and (w.l.o.g.) assume
that $c$ increased its potential in that round.  By Lemma~\ref{lem:knockout}, we conclude that in all
future rounds, $c$ has potential $\geq 1.5$ and $c'$ has potential $< 0$.  Thus $c'$ never enters the MIS.
\end{proof}

\paragraph{Definitions.}

The more interesting task is proving that eventually, every cell or one of its neighbors will enter the MIS,
and that this will happen quickly.  

We begin with a few definitions.  When a cell has no neighbors with larger potential, then it remain a 
candidate to enter the MIS.  We call such a cell a local maximum:
\begin{definition}
Fix some round $r$ and cell $c\in {\cal C}$. We say $c$ is a {\em local maximum} in $r$ if and only if for every 
neighbor $c'$ of $c$: $C_r(c) \geq C_r(c')$.  We say that $c$ is a {\em $k$-hop local maximum} if and only if for
every cell $c'$ within $k$ hops of $c$: $C_r(c) \geq C_r(c')$.
\end{definition}

Notice that if a cell is a local maximum, then it has no neighbors of larger potential, but it may well have neighbors of equal
potential.  In fact, it is only these neighbors of equal potential that will compete with it to enter the MIS.  We define $pN(c)$
to be exactly these neighbors, and the $p$-degree to be the number of such neighbors:
\begin{definition} 
Define $pN(c) = \{c' \in N(c) | C_r(c) = C_r(c')\}$, i.e., the active neighbors of $c$ that have the same potential.  
Define the $p$-degree $pDeg(c) = |pN(c)|$ to be the number of active neighbors of $c$ with the same potential as $c$.
\end{definition}

We will show that a cell $c$ that is a local maximum has (approximately) probability $1/pDeg(c)$ of entering the MIS within $O(\log \Delta)$ rounds.  
We will want to identify cells that are likely going to enter the MIS quickly, or have a neighbor that is likely to enter the MIS quickly.  We define a \emph{quick-entry} cell as follows:
\begin{definition}
We say that cell $c$ is a \emph{quick-entry} cell in round $r$ if it satisfies the following properties:
\begin{itemize}
	\item Cell $c$ is active.
	\item Cell $c$ is a local maximum.
	\item Every cell $c' \in pN(c)$ is a local maximum.
	\item For every cell $c' \in pN(c): pDeg(c') \leq 2\cdot pDeg(c)$.
\end{itemize}
\end{definition}
If $c$ is a quick-entry cell, we can argue that either it or one of its neighbors in $pN(c)$ will enter the MIS with constant probability.  
Because $c$ and its neighbors in $pN(c)$ are local maxima, we can expect they will enter the MIS with probability inversely propertional to their $p$-degree.  And because every neighbor of $c$ has $p$-degree at most twice that of $c$, we can conclude that each of these cells in $pN(c)$ has (approximately) probability $\geq 1/2\cdot pDeg(c)$ of entering the MIS.  Since there are $pDeg(c)$ such cells, we can conclude (after sidestepping issues of independence) that with constant probability, exactly one of them will enter the MIS.

\paragraph{Finding a quick-entry cell.}  
Our next step, then, is to show that for any cell $c'$, we can always find a quick-entry cell that is no more than $O(\log \Delta)$ hops away:
\begin{lemma}
\label{lem:find-quick}
Consider the subgraph consisting only of active cells.  
For every active cell $c'$, for every round $r$, there exists a quick-entry cell $c$ within distance $O(\log \Delta)$.
\end{lemma}
\begin{proof}
Fix a round $r$ and a cell $c'$.  We prove this lemma constructively.  We first identify a cell $c_1$ that is a $(\log(\Delta)+2)$-hop local maximum, meaning that no cell within $\log(\Delta)+2$ hops has larger potential.

Begin at cell $c_1 = c'$ and repeatedly execute the following:  If any cell $c$ within distance $\log{\Delta}+2$  (in the subgraph of active cells) has potential (strictly) larger than $c_1$, then set $c_1 = c$.  Repeat this process until every cell $c$ within distance $\log{\Delta}+2$  has $C_r(c) \leq C_r(c_1)$.  Since potentials change by $0.5$, the minimum potential is $-3$, and the maximum potential is $2$, it is clear we can only repeat this procedure at most 10 times before finding a cell with potential $2$ (i.e., the maximum potential).  Thus within distance $10(\log(\Delta) + 2)$ of our starting cell $c'$, we have identified an active cell $c_1$ that is a $(\log(\Delta)+1)$-hop local maximum.

Next, we identify a cell $c_2$ that has the same potential as $c_1$, is also a local maximum, its neighbors in $pN(c_2)$ are also local maxima, and it satisfies the requisite degree property.  Begin at cell $c_2 = c_1$.  Repeatedly execute the following: If any cell $c$ in $pN(c_2)$ has $p$-degree $pDeg(c) \geq 2\cdot pDeg(c_2)$, then set $c_2 = c$.  (Again, recall that we are only considering neighbors in the subgraph of active cells.) Repeat this process until every neighbor of $c_2$ in $pN(c_2)$ has $p$-degree less than $2\cdot pDeg(c_2)$.  

Notice that this process must terminate, since at each step $pDeg(c_2)$ doubles, and the $p$-degree can never be larger than $\Delta$.  Thus, within $\log{\Delta}$ steps, we have found such a cell $c_2$.  At each step, we moved from cell $c_2$ to a neighbor in $pN(c_2)$, so the potential of $c_2$ remains equal to the potential of $c_1$ throughout the process.  Finally, since $c_1$ was a $(\log(\Delta)+2)$-hop local maximum, and since $c_2$ is at most $\log(\Delta)$ hops away from $c_1$, we know that no cell within $2$ hops of $c_2$ has a larger potential than $c_2$ or its neighbors in $pN(c_2)$.  Thus, we conclude that $c_2$ and the cells in $pN(c_2)$ are all local maximum.  

Thus cell $c_2$ is a quick-entry cell, and $c_2$ is within distance $O(\log(\Delta))$ of the initial cell $c'$.
\end{proof}

\paragraph{Probability of quick-entry.}
Next, we prove that, given a quick-entry cell $c$, either $c$ or one of its neighbors really does enter the MIS with constant probability within $O(\log(\Delta))$ rounds.
\begin{lemma}
\label{lem:joinmis}
Consider the subgraph consisting only of active cells.
Let $c$ be a quick-entry cell. Then with probability at least $1/16$, either $c$ or a neighbor of $c$ enters the MIS within $O(\log \Delta)$ rounds.
\end{lemma}
\begin{proof}
Let $S = \{c\} \cup pN(c)$, and let $s = |pN(c)|$.  Notice every cell in $S$ has $p$-degree at most $2s$ in round $r$, and recall that every cell in $S$ is a local maximum in round $r$.  

In the special case where $c$ has no neighbors with the same potential, then $c$ enters the MIS with constant probability: if its potential is $< 0.5$, it advances (neither sending nor receiving) until its potential is $0.5$; at that point, with probability $1/2$ it broadcasts and advances its potential to $1.5$ and enters the MIS in the following round.  Since every neighbor has potential $< 0.5$ during this process, none of them broadcast and interfere.  We assume in the following that $c$ has at least one neighbor in round $r$ in $pN(c)$, i.e., $s \geq 1$.

In every round, we update $S$ as follows: if $c' \in S$ is a cell in $S$, and if the current round is a competition round for $c'$ in which $c'$ does not broadcast, then we remove $c'$ from $S$.  Intuitively, $S$ is the set of cells that remain candidates for entering the MIS.

Throughout the analysis, we can ignore inactive cells.  None of the cells in question neighbor an inactive cell with potential $2$ (as then they would themselves be inactive), and the remaining inactive cells have potential $< 0$ and hence never broadcast.

We begin with some useful observations about the behavior of a cell $c' \in S$ that is a local maximum.  We will then leverage these to make our main probabilistic argument.

\paragraph{Observations.} 
First, notice that as long as a cell $c'$ remains in $S$, it remains a local maximum: If $c'$ has potential $< 0.5$, neither $c'$ or its neighbors broadcast, so $c'$ increases its potential by $0.5$, while its neighbors increase their potential by at most $0.5$.  If $c'$ has potential $\geq 0.5$, then it has to broadcast and so it decreases its potential by at most $1/2$, while its neighbor also decrease their potential by at least $1/2$.  (If $c'$ has potential exactly $0.5$, it has to broadcast to remain in $S$.)  In either case, the potential gap between $c'$ and its neighbors is preserved.

Second, once a cell $c' \in S$ arrives at a competition round for the first time in the period starting with round $r$, it will return to a competition round every other round until either: (a) it is silent in a competition round in which a neighbor sends, or (b) it sends alone in a competition round in its neighborhood.  As we already proved in Lemma~\ref{lem:knockout}, if cell $c' \in S$ sends in a competition round $r'$, and no neighbor of $c'$ sends in $r'$, then no neighbor will ever send again and $c'$ will join the MIS in the next round.  Thus in case (b), the cell $c'$ joins the MIS within one round.

If, in round $r$, the cells in $S$ have potential 1, then in round $r$ they all broadcast, they all receive, and they all lower their potential to $0.5$.  If, in round $r$, the cells in $S$ have potential $< 0.5$, then they all advance in lockstep, increasing their potential until it is $0.5$.  Thus, all the cells in $S$ enter a competition round for the first time in the same round.

From that point on, cells in $S$ always remain at the same potential in every round, entering competition rounds every other round until a cell is silent in a competition round (in which case it exits $S$) or a cell in $S$ sends alone in its neighborhood in a competition round (in which case it enters the MIS). We refer to the rounds in which cells in $S$ have potential $0.5$ as $S$-competition rounds.

%

To analyze the behavior of cells during these competition rounds, we imagine that each cell $c'$ has a random bit string $B_{c'}$ of length $\log(4s)$ where each bit is $0$ with probability $1/2$ and $1$ with probability $1/2$.  Whenever cell $c'$ is in a competition round, it uses the next unused bit from the string $B_{c'}$ to decide whether or not to broadcast.  (This allows us to analyze the possible behavior of cells throughout the $S$-competition rounds, even if some cell enters the MIS early.)   

\paragraph{Winners and losers.}
For each cell $c' \in S$, we say that $c'$ is a \emph{winner} if: (i) the $B_{c'}$ string contains all $1$'s (i.e., it broadcasts in all the $S$-competition rounds), and (ii) for every cell $c'' \in pN(c') \setminus S$, the string $B_{c''}$ has at least one $0$ (i.e., it does \emph{not} broadcast in all the $S$-competition rounds).  Intuitively, this means that excluding cells in $S$, we can be sure that within $\log(4s)$ rounds, cell $c'$ knocks out each of its neighbors and hence enters the MIS.  (Obviously if $c'$ wins early, it may not consume all of its bits in $B_{c'}$.)

Notice that the cells in $pN(c') \setminus S$ are the neighbors of $c'$ that are not in $S$, but that also have potential $0.5$ during rounds in $T$.  There are at most $2s$ such neighbors, by the definition of $S$.  We do not care about the other neighbors of $c'$: since $c'$ is a local maximum, they all have potential $\leq 0$ and hence do not broadcast and cannot prevent $c'$ from entering the MIS (as long as $c'$ continues to broadcast in rounds in $T$).

We now argue that the probability that a cell $c'$ is a winner is $\geq 1/(8s)$.  First, each bit in $B_{c'}$ is $1$ with probability $1/2$, and thus the probability that the bit string is all $1$'s is $1/2^{\log(4s)} = 1/(4s)$.  Second, the same holds for each cell in $pN(c') \setminus S$, i.e., it has a probability $1/(4s)$ of having all $1$'s in its random bit string.  By a union bound, the probability that any cell in $pN(c') \setminus S$ has all $1$'s in its bit string is $\leq 2s/(4s) \leq 1/2$, since there are never more than $2s$ cells in this set of neighbors.  Therefore, the probability that no cell in $pN(c') \setminus S$ has all $1$'s in its bit string is at least $1/2$.  Since the bit strings are independent, the probability that $c'$ is a winner is at least $1/(8s)$.

We define the event $W_{c'}$ to be the event that (i) cell $c'$ is a winner and (ii) no other cell in $S$ has an all $1$'s bit string (i.e., no other cell in $S$ broadcasts in all the $S$-competition rounds).  Notice that for any two cells $c'$ and $c''$, the events $W_{c'}$ and $W_{c''}$ are disjoint since if they were both winners, then they would both have all $1$'s bit strings (and would both be broadcasting in every round in $T$).   Also, notice that the probability that $c'$ is a winner is independent of the probability that any other cell in $S$ has an all $1$'s bit string (as it only depends on neighbors not in $S$).

The probability of $W_{c'}$ can be bounded as follows. 
\begin{eqnarray*}
\Pr{W_{c'}} = \Pr{c' \textrm{ is a winner}}(1 - 1/2^{\log(4s)})^s \\
& \geq & \frac{1}{8s}\left(1 - \frac{1}{4s}\right)^s \\
& \geq & \frac{1}{8s}e^{-1/2} \\
& \geq & \frac{1}{16s}
\end{eqnarray*}
Thus, summing over the disjoint events $W_{c'}$ for all $c' \in S$, we conclude that the event $W_{c'}$ occurs for one cell in $S$ with probability at least $1/16$.  

This implies that, with probability at least $1/16$, by the end of the $S$-competition rounds, there is exactly one cell $c'$ in $S$ that is a winner.  This winner cell $c'$ necessarily goes on to enter the MIS within $1$ further rounds, as it has successfully suppressed all of its neighbors.  Since this occurs within $\log(4s)$ $S$-competition rounds, and these competition rounds alternate (until there is a winner), we conclude that this occurs within time $O(\log \Delta)$.
\end{proof}

Putting together the previous two lemmas, we conclude the following:
\begin{lemma}
\label{lem:smis:nearby}
Given any cell $c$ and round $r$, with probability at least $1/16$ there is a cell within distance $O(\log \Delta)$ that enters the MIS within $O(\log \Delta)$ rounds.
\end{lemma}
\begin{proof}
Lemma~\ref{lem:find-quick} inplies that there is a quick-entry cell $c'$ within distance $O(\log \Delta)$ of $c$, and Lemma~\ref{lem:joinmis} says that either $c'$ or a neighbor of $c'$ enters the MIS within $O(\log \Delta)$ rounds with probability at least $1/16$.
\end{proof}

\paragraph{Leveraging topology.}

We are now ready to prove that eventually every cell or one of its neighbors enters the MIS, and that this occurs quickly.  To prove this, we will liverage the assumption that the underlying graph topology $G=(V,E)$ is a UBG with constant doubling dimension.  The key property we need from the topology of the graph (which is a standard property) is as follows:
\begin{lemma}
\label{lem:ubg}
For every independent set $I$, for every cell $c$: there are at most $O(k^{\rho})$ cells in $I$ within distance $k$ of $c$.
\end{lemma}
\begin{proof}
Fix an independent set $I$ and a cell $c$.  Let $\rho$ be the doubling dimension of the graph.  Recall that $c$ is embedded into a metric space, and hence we can consider the ball $B$ of radius $k$ around $c$.  This ball contains all the cells within $k$ hops of $c$ because neighbors in the graph have distance at most $1$.  

This ball $B$ is covered by at most $(2k)^{\rho}$ balls of radius $1/2$.  This can be seen inductively: for the base case, a ball of radius $1$ is covered by at most $2^{\rho}$ balls of radius $1/2$; for the inductive step, a ball of radius $k$ is covered by at most $2^{\rho}$ balls of radius $k/2$, each of which is covered by at most $(2(k/2))^{\rho}$ balls of radius $1/2$, i.e., $(2k)^{\rho}$ balls of radius $1/2$ in total.

Let $S$ be a set of at most $(2k)^{\rho}$ balls of radius $1/2$ that cover $B$.  Notice that each of these balls can contain at most one cell in the independent set $I$: any two cells in a ball of radius $1/2$ are within distance $1$ of each other and hence must be neighbors.  Each cell in $I$ that is within distance $k$ of $c$ lies in one of the balls in $S$, and to there are at most $(2k)^{\rho} = O(k^{\rho})$ such cells.
\end{proof}

We can now show that for every cell $c$, it either joins the MIS or has a neighbor join the MIS within a fixed amount of time:
\begin{lemma}
\label{lem:liveness}
Consider a network of $n \geq 1$ {\tt KnockBack} cells connected in a unit ball graph $G$ with constant doubling dimension and maximum degree $\Delta$. 
For any $\epsilon > 0$, for every cell $c$: within time $O(\polylog(\Delta)\log(1/\epsilon)$, either $c$ or a neighbor of $c$ enters the MIS.
\end{lemma}
\begin{proof}
%

First, partition rounds into {\em phases} of length $t = O(\log{\Delta})$, where the constant is large enough for Lemma~\ref{lem:joinmis} to hold.  Label these epochs $1,2,3, \ldots$, and so on. 
By Lemma~\ref{lem:smis:nearby}, for each phase $i$, if $c$ is active at the start of $i$, then with probability $p \geq 1/16$, an active cell $c'$ within
$O(\log \Delta)$ hops of $c$ joins the MIS in this phase. We call a phase {\em successful} with respect to $c$ if this event occurs.

By Lemma~\ref{lem:ubg}, we know that there can be at most $\gamma = O(\polylog(\Delta))$ cells that join the MIS within distance $O(\log{\Delta})$ of $c$.  Thus there can be at most $\gamma$ successful phases before $c$ becomes inactive.

Let random variable $X_i$ be a random indicator variable that evaluates
to $1$ with probability $1/16$, and otherwise $0$. Let $Y_k = \sum_{i=1}^k X_i$.
Let $p_k$ be the probability that we have at least $\gamma$ successful phases with respect
to $c$ in the first $k$ phases.
A standard stochastic dominance
 argument establishes that $p_k > \Pr(Y_k \geq \gamma)$.

Consider $k= 128\gamma \ln(1/\epsilon)$.  For this value, $E(Y_k) = 8 \cdot \gamma\ln(1/\epsilon)$.
Applying a Chernoff bound, it follows:

\begin{eqnarray*}
\Pr(Y_k \leq \gamma) & \leq & \Pr(Y_k \leq E(Y_k)/2) \\
& \leq & e^{-E(Y_k)/8} \\
& \leq & e^{-8\gamma \ln(1/\epsilon)/8} \\
& \leq & \epsilon
\end{eqnarray*}

Therefore, the probability that $c$ remains active for more than $k$ phases
is less than $\epsilon$.  To conclude the proof, we note that there are $\Theta(\log{\Delta})$
rounds per phase, and $k = \Theta(\log(\Delta)\ln{n})$,
yielding $O(\polylog{\Delta}\ln{n})$ rounds.
\end{proof}

Putting these pieces together, we get the main theorem:
\begin{theorem}
Consider a network of $n \geq 1$ {\tt KnockBack} cells connected in a unit ball graph $G$ with constant doubling dimension and maximum degree $\Delta$.  With
probability at least $1 - 1/n$, every cell in inactive within time $O(\polylog(\Delta)\log(n))$, and the resulting cells with potential $2$ form an MIS.
\end{theorem}
\begin{proof}
Setting $\epsilon = 1/n^2$, we conclude from Lemma~\ref{lem:liveness} that a cell becomes inactive within $O(\polylog(\Delta)\log(n))$ rounds with probability at least $1 - 1/n^2$.  Taking a union bound, we conclude that every cell is inactive within $O(\polylog(\Delta)\log(n))$ rounds with probability at least $1 - 1/n$.  By Lemma~\ref{lem:safety} we know that no two neighboring cells join the MIS, and so we conclude that the result is a correct MIS.
\end{proof}

\paragraph{Stabilization.} Throughout the analysis above, we assumed for simplicity that all the cells began with potential precisely zero.  However, it turns out that is not in fact necessary.  Here we discuss the behavior of a network of cells when they begin with arbitrary potentials.

The first issue to address is that throughout the analysis above, we discussed fixed quantized potentials: 0, 0.5, 1, 1.5, etc.  Starting with a potential of zero, it is impossible to arrive at a potential that is not a multiple of 0.5.  This quantization, however, is not important---what matters is the range that the potential is in.  All cells with potential in the range $[0.5, 1)$ act identically, as do cells with potential in the range $[1, 1.5)$, etc.  Thus, for the purpose of analysis, we can treat the potential as ``rounded down'' to the nearest multiple of $0.5$.  (For example, when deciding if a cell with potential 0.7 is a local maximum, we round its potential down to 0.5 and compare its potential accordingly.)

The second issue to address is that potentials may begin too low, e.g., $< -3$.  There are basically two ways that this can resolve itself: either a neighbor enters the MIS, or eventually the potential climbs into the normal range (due to the gradient effect).  In either case, the situation eventually resolves itself.

The third issue is that potentials may begin too high.  Multiple cells may begin with potential 2, i.e., part of the MIS (or with potential even higher, if that is feasible in the system).  Anytime there are two neighbors with potential $> 1$, they will continue to broadcast in every round and hence eventually one or both will exit the MIS, with their potential dropping below $2$.  Once safety has been restored, i.e., no neighbors are in the MIS, then the system will stabilize as already described.

Finally, observe that nowhere in the analysis did we depend on any special initial conditions or relations between the potentials.  The key idea was to identify local maxima, and argue that they have a reasonable probability of entering the MIS.  Nothing about the initial conditions matter here.  Thus we conclude:

\begin{theorem}
Consider a network of $n \geq 1$ {\tt KnockBack} cells connected in a unit ball graph $G$ with constant doubling dimension and maximum degree $\delta$.  Assume that the cells begin with arbitrary potentials.  Then eventually, with probability 1: no two neighboring cells are in the MIS, and every cell is either in the MIS or has a neighbor in the MIS.
\end{theorem}

%% file: type.tex
\section{Input Type Computation}
\label{sec:type}

We previously studied the ability of a collection of identical cells to
stochastically break symmetry.
We now turn our attention to the ability of cells to process information.
There are different natural definitions for information processing in this setting.
We begin here with a definition studied in the
bio-inspired chemical reaction network and population protocol settings: computation
on input type counts.

In more detail, for these problems
the {\em input} to the computation is the {\em a priori} unknown counts of the different cell types in
the system.
We look at two common problems from the input type computation literature (see Section~\ref{sec:intro}): threshold
and majority detection.
We produce probabilistic approximate solutions in both cases,
and prove for threshold detection a lower bound that says a non-zero error probability is necessary
for general solutions, even when considering approximate versions of the problem.
This bound is easily adapted to majority detection as well,
implying that although these problems are tractable in the CBM,
they require randomized solutions.

\subsection{Threshold Detection}
\label{sec:type:threshold}

We begin with the threshold detection problem.
The goal is to count if the number of cells of a designated type is beyond some fixed threshold.
We work with the following general definition that allows us to study various problem parameters.

\begin{definition}
The {\em $(k, \tau, \epsilon)$-threshold detection} problem is parameterized with an integer threshold $k\geq 1$,
an integer threshold range  $\tau \geq 1$, and a fractional error bound  $\epsilon > 0$.
A cell definition {\em solves} the $(k, \tau, \epsilon)$-threshold detection problem in $T$ rounds,
if it guarantees for any network size $n \geq 1$, that
when $n$ copies of the cell are run in a fully connected network topology,
by round $T$ the following is true with probability at least $1-\epsilon$:
\begin{itemize}
	\item If $n > \tau\cdot k$, at least one cell executes a {\em threshold exceeded} expression event.
	\item If $n \leq k/\tau$, no cell executes a {\em threshold exceeded} event. 
\end{itemize}
\label{def:threshold}
\end{definition}

We say a cell type solving this problem is {\em deterministic} if and only if the probabilities in the codomain
of the cell firing functions are $0$ and $1$.
The traditional (and strictest) definition of this problem assumes deterministic solutions for $\tau = 1$ and $\epsilon = 0$, that is,
with no probability of error, it detects if the cell count is greater than $k$.

Below, we begin with a straightforward one round deterministic solution for the traditional
parameters $\tau =1$ and $\epsilon = 0$, that works only if the binding
constant for the cell is at least $k$ (recall, as defined in Section~\ref{sec:model},
the binding constant is
the maximum number of incoming ligand counts the membrane function can distinguish).
We then tackle the case of general $k$ values, and present a solution
that for any error bound $\epsilon >0$, solves the $(k, 8\ln{(1/\epsilon)}, \epsilon)$-threshold detection problem.
Finally, we prove that the assumption of a non-zero error probability is necessary to solve this problem for arbitrary $k$.


\subsubsection{Deterministic Solution for Small Thresholds}
We begin with the easy case where  $k \leq b$,
where $b$ is the maximum binding constant we are allowed to use in defining our cell.
In this case, the cell's membrane function can distinguish between $1, 2, ..., k$ {\em or more},
incoming ligands of a given type. By simply having every cell send a ligand in round $1$
with probability $1$, 
all cells can directly count if there are $k$ or more other cells in the system.
We capture this strategy with the following {\tt SmallThreshold}$(k)$ cell definition,
which is parameterized with the threshold $k$ it is tasked with detecting:

\begin{center}
\begin{tabular}{|cc|}
\hline
\multicolumn{2}{|c|}{{\tt SmallThreshold}$(k)$ cell definition} \\
\hline
\hline
$q_0 = 1$ & ${\cal B} = \{(f, (0, m))\}$, where: \\
$\lambda = 1$, $\sigma= 0$ & \multicolumn{1}{r|}{$f(x \geq 1) = 1$}\\
%
$g(|M| \geq k) = 2$ & \multicolumn{1}{r|}{$f(x < 1) = 0$}\\
$g(|M| < k) = 0$ & \\
\hline
\multicolumn{2}{|c|}{threshold exceeded event threshold: $2$}\\
\hline
\end{tabular}
\end{center}

\noindent It is straightforward to establish the correctness of this strategy:

\begin{theorem}
Fix any threshold $k \geq 1$ that is smaller than the maximum allowable cell binding constant.
The deterministic {\tt SmallThreshold}$(k)$ cell definition solves the $(k, 0, 0)$-threshold detection problem
in one round.
\end{theorem}
\begin{proof}
The {\tt SmallThreshold} definition starts each cell at potential value $1$. During the first round, all cells fire their single event with probability $1$.
We consider two cases for $n$, the number of cells in the single hop network.
 If $n > k$, then all cells receive at least $k$ ligands (i.e., one ligand from the $n-1 \geq k$ other cells).
 Accordingly, the membrane function will add $2$ to each cell's potential. 
 Because the equilibrium $\sigma = 0$ and the gradient rate $\lambda=1$, 
 each cell will also reduce its potential by $1$ as ion flux pushes it back toward the equilibrium.
 The result in this case is that each cell ends up with potential $1+2-1 = 2$, correctly triggering the expression event.
 On the other hand, if $n \leq k$, then the membrane function does not impact the potential values,
 and each cell ends up at the equilibrium potential $1-1 = 0$. Because the firing probability is $0$ at this potential
 value, the cells will  correctly remain at this potential forever.
 \end{proof}

\subsubsection{Probabilistic Solution for Arbitrary Thresholds}
We now consider the general case where $k$ might be large compared to the cell binding constant.
We can no longer directly count the number of neighbors using the membrane function.
If we tolerate a non-zero error probability and some looseness in the threshold range (as defined by $\tau$), 
we can solve this problem probabilistically, trading off increases in $\tau$ for decreases in error.

Here we study  {\tt GeneralThreshold}$(k)$,
a natural generalization of {\tt SmallThreshold}$(k)$.
With this new definition, instead of having all cells fire with probability $1$
in the first round, we have them instead fire with probability $1/k$. 
If {\em any} cell fires, it moves itself past the event threshold, otherwise, the system falls back to a quiescent equilibrium.
In more detail:

\begin{center}
\begin{tabular}{|cc|}
\hline
\multicolumn{2}{|c|}{{\tt GeneralThreshold}$(k)$ cell definition} \\
\hline
\hline
$q_0 = 1$ & ${\cal B} = \{(f, (2, m))\}$, where: \\
$\lambda = 1$, $\sigma= 0$ & \multicolumn{1}{r|}{$f(x \geq 1) = 1/k$}\\
%
$g(|M| \geq 1) = 2$ &  \multicolumn{1}{r|}{$f(x < 1) = 0$} \\
$g(|M| < 1) = 0$ & \\
\hline
\multicolumn{2}{|c|}{event threshold: $2$}\\
\hline
\end{tabular}
\end{center}

\noindent It is obvious that {\tt GeneralThreshold}$(k)$ exhibits generally correct behavior: if $n$ is at least $k$, then we expect
cells to fire in the first round, and if $n$ is sufficiently small relative to $k$, then we do not expect any cell to fire in the first round.
The following theorem more precisely captures the trade-off between the threshold range and the error probability:

\begin{theorem}
Fix any error bound $\epsilon, 0 < \epsilon < 1$ and threshold $k \geq 1$.
The {\tt GeneralThreshold}$(k)$ cell definition solves the $(k, 8\ln{(1/\epsilon)}, \epsilon)$-threshold detection problem
in one round. 
\end{theorem}
\begin{proof}
Fix some error bound $\epsilon$, threshold $k$, and network size $n$.
To simplify notation, let $t = 8\ln{(1/\epsilon)}$.
Fix some random indicator variable for each cell that evaluates to $1$ if that cell fires in the first round and otherwise evaluates to $0$. 
Let $X$  the sum of these random variables.
If $X \geq 1$, all cells will execute the relevant expression event, otherwise, no cell will ever execute this event.

We first consider the case that $n > tk$. In this case, we require the expression event is executed. Therefore, we will show that
if $n > tk$ then $X \geq 1$ with probability at least $1-\epsilon$.
Because $X$ is the sum of independent random indicator variables,
and $E(X) > t$,
we can establish this claim with the following form of a Chernoff bound:

\[\Pr[ X \leq (1 - \delta)E(X)]  \leq  \exp\big( - (\delta^2 (E(X)/2)   \big) = \exp\big( - (\delta^2 4\ln{(1/\epsilon)} \big).\]

\noindent Notice, for $\delta = 1/2$, $(1-\delta)E(X) > 1$, and the error probability is upper bounded by $\exp(-\ln{(1/\epsilon)}) = \epsilon$,
as required.
Now consider the case where $n \leq k/t$. Here we require that the express event is {\em not} executed.
Therefore, we must show that if $n \leq k/t$, then $X < 1 \Rightarrow X = 0$, with sufficient probability.
For this case, $E(X) \leq 1/t$.
We now apply the following upper bound version of a Chernoff bound:

\[  \Pr[X \geq (1+\delta)E(X) ] \leq \exp\big(  -\frac{\delta^2 E(X)}{3} \big).\]

\noindent If we set $\delta = \sqrt{24}\ln{(1/\epsilon)}$, it follows $(1+\delta) < (1+\sqrt{24})\ln{(1/\epsilon)}$, therefore:

\[(1+\delta)E(X) \leq ((1+\sqrt{24}) \ln{(1/\epsilon)})E(X) < (6\ln{(1/\epsilon)})/(8\ln{(1/\epsilon)}) < 1.\]

If we then plug this value of $\delta$ into our above Chernoff bound form,
it follows that the probability that $X$ is at least as large as the value we bounded above to be less than $1$,
is upper bounded by:

 \[ \exp\big(  -\frac{\delta^2 E(X)}{3} \big) \leq \exp\big(  -\frac{24\ln^2{(1/\epsilon)}}{3t} \big) =\exp\big(  -\frac{24\ln^2{(1/\epsilon)}}{24\ln{(1/\epsilon)}} \big)= \epsilon,\]

\noindent as required.
\end{proof}

A nice property of this solution is that the threshold range is a function only of $\tau$ (and not $k$ or $n$). 
Almost certainly, a more complex strategy could provide a smaller $\tau$ bound for a given $\epsilon$ (perhaps, for example,
by requiring a certain fractions of a fixed length of rounds to include event firings).
As with our study of symmetry breaking, however, we are prioritizing the bio-plausible simplicity
 of {\tt GeneralThreshold}.
 Another potential improvement would be removing the non-zero error bound. In the next section, we prove that optimization impossible.

\subsubsection{Lower Bounds for Threshold Detection}

Our {\tt SmallThreshold} solved the problem with $\epsilon =0$ and a deterministic definition.
It required, however, that the binding bound used in the cell definition was large with respect to the threshold $k$.
Here we prove two results that establish without this condition neither $\epsilon = 0$ nor determinism are possible.
This justifies in some sense the use of probabilistic firing functions and an error bound $\epsilon  > 0$ in 
{\tt GeneralThreshold}.

We begin by establishing the shortcomings of determinism.

\begin{theorem}
Fix a binding bound $b \geq 1$, threshold range $\tau \geq 1$,  non-trivial error probability $\epsilon$, $0 \leq \epsilon < 1/2$, and round length $T \geq 1$.
There does not exist a deterministic cell definition with binding bound $b$ that solves the $(k, \tau, \epsilon)$-threshold detection problem in $T$ rounds
for every threshold $k \geq 1$.
\label{thm:lower:threshold1}
\end{theorem} 
\begin{proof}
Assume for contradiction there exists a $b$, $\tau$, $\epsilon$, $T$ and cell definition that solves the problem for these parameters for every $k$.
Because the cell is deterministic, a straightforward inductive argument establishes that if all cells start with the same potential then they will have the 
same potential after each round. This follows because if any cell fires an event in a given round, then all cells fire the same event, and therefore,
all cells receive the same number of ligands associated with the event. The resulting offsets to their potential due to firing and the membrane functions are the same.

Consider the behavior of this cell for threshold $k > \tau\cdot (b+1)$ and $n > \tau k$. With probability at least $1-\epsilon > 1/2$,
this system will correctly arrive at a state in the first $T$ rounds in which a cell (and therefore every cell) executes the threshold exceeded event.
Now consider the behavior of this cell for $n=k/\tau - 1$. Notice, because we assumed $k > \tau\cdot (b+1)$, $k/\tau - 1 > b$.
In both the $n > \tau k$ and $n = k\tau - 1$ case, whenever any event fires, all cells will receive more than $b$ copies,
meaning their membrane functions evaluate the {\em $b$ or more case}.
Therefore, in both cases, the two systems will behave the same, meaning that in the second case,
the system will {\em incorrectly} execute the expression even with probability $1-\epsilon > 1/2$, contradicting the assumption of an error bound less than $1/2$.
\end{proof}

\noindent Building on the above result, we now establish the shortcomings of zero error probability, even if we consider non-deterministic solutions:

\begin{theorem}
Fix a binding bound $b \geq 1$, threshold range $\tau \geq 1$, and round length $T \geq 1$.
There does not exist a cell definition with binding bound $b$ that solves the $(k, \tau, 0)$-threshold detection problem in $T$ rounds
for every threshold $k \geq 1$.
\label{thm:lower:threshold2}
\end{theorem} 
\begin{proof}
Assume for contradiction there exists a $b$, $\tau$, $\epsilon$, $T$ and cell definition $C$ that solves the problem for these parameters for every $k$.
It follows that for every $k$ and $n$,
every possible probabilistic execution of $C$ (parameterized with $k$) leads to the correct behavior (because the assumed error probability is $0$).
Of the various possible executions is the one in which every time a firing function had a non-zero probability of firing, it fired.
This execution always behaves the same as the deterministic cell $C'$ that transforms every firing function from $C$ into a monotonic unit step function
in which all probabilities greater than $0$ are increased to $1$.
Therefore, this deterministic cell solves the problem with $\epsilon = 0$ for all $k$.
By Theorem~\ref{thm:lower:threshold2}, however,
there must exist some $k$ values for which $C'$ does not satisfy this error bound. A contradiction.
\end{proof}

\subsection{Majority Detection}
\label{sec:type:majority}

The majority detection problem requires the algorithm designer
to provide two cell definitions: one for cell type $A$ and one for cell type $B$.
Each type has its own expression event that suppresses the other cell type.
An arbitrary number of $A$ and $B$ cells are connected in a fully-connected network. 
We use $n_A$ and $n_B$ to indicate these counts, respectively.
If $n_A > n_B$, the correct outcome is for at least one type $A$ cell to execute its expression
event before any type $B$ cell. If $n_B > n_A$, then the correct outcome is the reverse.
If $n_A = n_B$, any outcome is correct.

Solving this problem correctly when $n_A$ and $n_B$ are very close presents a difficult task.
Following the lead of our study of the threshold detection problem,
and the study of majority detection in other models (e.g.,~\cite{AngluinAE08}),
we might instead require an {\em approximate} solution that is required
to be correct only when the gap between the two counts is sufficiently
large. The goal is to identify a strategy that offers a reasonable trade-off between the gap
required for a fixed error bound.
We tackle this challenge with the below cell definition
which is parameterized with an upper bound $N$ on the maximum network size and a constant error bound $\epsilon > 0$:

\begin{center}
\begin{tabular}{|cc|}
\hline
\multicolumn{2}{|c|}{{\tt MajorityA}$(N,\alpha = \lceil 2\ln{(2/\epsilon)} \rceil)$ cell definition (for type $A$)} \\
\hline
\hline
$q_0 = 0$ & ${\cal B} = \{(f, (\alpha\log{N}, m_A))\}$, where: \\
$\lambda = 1$, $\sigma= 3\alpha\log{N}$ & \multicolumn{1}{c|}{$f(0 \leq x \leq \alpha \log{N}) = 2^{-  (\log{N} - \lfloor \frac{x}{\alpha} \rfloor) }$}\\
%
$g(|M_B| \geq 1) = -2\alpha\log{N}$ &  \multicolumn{1}{c|}{$f(x < 0) = 0$} \\
$g(|M_B| = 0) = 0$ & \multicolumn{1}{c|}{$f(x > \alpha\log{N}) = 1$}  \\
\hline
\multicolumn{2}{|c|}{event threshold: $3\alpha\log{N}$}\\
\hline
\multicolumn{2}{|c|}{($M_B$ equals the sub-multiset including only ligands of type $m_B$ sent from type $B$ cells.)}\\
\hline
\end{tabular}
\end{center}


We prove the following regarding the trade-off between $\epsilon$ and the required size gap between the cell type counts,
which, roughly speaking, shows that for any constant error bound,
there is a constant size gap between $n_A$ and $n_B$ which will enable majority detection to work with that error.

\begin{theorem}
Fix some constant error bound $\epsilon > 0$ and upper bound $N>1$.
Let $\alpha = \lceil 2\ln{(2/\epsilon)} \rceil$.
The {\tt MajorityA}$(N,\alpha)$ and {\tt MajorityB}($N, \alpha)$ cell definitions, when executed in a system with $n_A$ and $n_B$ type $A$ and type $B$ cells, respectively,
where $n_A > n_B \cdot (\alpha 4)/\epsilon$ and $N \geq n_A + n_B$, guarantees with probability at least $1-\epsilon$: in the first $O(\log{n})$ rounds, a type $A$ expression event
will occur before any type $B$ event.
(The symmetric claim also holds for $n_B > n_A \cdot (\alpha 4)/\epsilon$.)
\end{theorem}
\begin{proof}
Fix some $\epsilon$, $\alpha= \lceil 2\ln{(2/\epsilon)} \rceil$, $N$, $n_A$ and $n_B$ such that $n_A > n_B \cdot (\alpha 4)/\epsilon$.
The argument for the case where $n_B$ is larger is symmetric.
In the absence of a bioelectric event firing,
the cells spend $\alpha$ rounds at each firing probability in the geometric series $1/N, 2/N, ..., 1/2$.
We call the $\alpha$ rounds spent at each probability a {\em phase}.
Consider the first phase associated with a probability $p \geq 1/(2n_A)$.
By our assumption, $p < 1/(\alpha q n_B)$, where $q=(4/\epsilon)$.

By a union bound, the probability that all $n_B$ cells of type $B$ are quiet during all $\alpha$ rounds of the phase associated with probability $p$
is less than $\alpha n_B p < 1/q$.
Also by a union bound, the probability that all type $B$ cells are silent in all phases up to and including this phase is upper bounded by:

\[ 1/q + 1/(2q) + 1/(4q) + ... + 1/N < 1/q\sum_{i=0}^{\infty} 2^{-i}  \leq 2/q = \epsilon/2. \]

Now consider the negative event where {\em no} type $A$ cell fires before the end of this phase.
This probability is upper bounded by:

\[(1-p)^{\alpha n_A} \leq e^{-p \alpha n_A} < e^{-  (\alpha n_A)/(2 n_A)  } = e^{-\alpha/2} \leq e^{\ln{(2/\epsilon)}} = \epsilon/2. \]

By a final union bound, the probability at least one of these two negative events occurs is less than $\epsilon$.
If {\em neither} of these events occur, a type $A$ cell 
fires before any type $B$ cell.
By the cell definitions,
this event will cause all type $B$ cells to drop to potential values below $0$ (where they have probability $0$ of firing),
and the type $A$ cell(s) that fired to increase their potential to a point where they will subsequently fire with
probability $1$, suppressing the type $B$ cells until they cross the type $A$ expression threshold.
\end{proof}

%% file: tm.tex
\section{Turing Completeness}
\label{sec:tm}

In the previous section,
we considered a definition of information processing in which cells computed functions on the number of cells of specific types in the network.
Here we study another natural definition in which cells process information embedded into their initial state.
In particular, we consider the ability of cells to compute functions on the initial potential value of a designated {\em input cell} in the system.
This setup helps isolate a core question in studying bioelectrical systems: {\em What types of computations on cell state can be computed through simple
bioelectric interactions?} 
Here we prove a perhaps surprising answer: {\em Essentially all feasible computations.}
Formally:

\begin{theorem}
Fix an arbitrary deterministic TM $M$. 
There exists a finite collection of cells defined with respect to $M$,
including a designated {\em input cell},
such that for every TM input $w$,
if you set the input cell's initial potential value to a specified unary encoding of $w$,
the cells will correctly simulate $M$ on $w$.
\label{thm:tm}
\end{theorem}

A well-known result due to Minsky~\cite{minsky1967} is that a counter machine (CM) with a constant number of counters
can simulate an arbitrary deterministic Turing machine (TM). The simulation requires that the input $w$ is encoded into a unary value
that is stored as the initial value of a designated counter.
To prove our theorem, therefore, it is sufficient to describe a strategy for a finite number of cells connected in a fully-connected network
topology to simulate an arbitrary CM.  Our simulation includes one {\em counter cell} for each counter in the simulated machine.
The initial potential value of each counter cell is interpreted as the initial value of the counter in the simulated machine.

More generally, given a TM $M$ let $C'$ be the corresponding counter machine constructed using known results.
Without loss of generality, we can transform $C'$ into an equivalent {\em well-formed} counter machine $C$ (see definition below)
that satisfies some additional behavioral constraints that will simplify our implementation of $C$ in the CBM. 
This transformation increases the number of states by at most a constant factor.
To simulate $M$, therefore,
it is sufficient to simulate the corresponding well-formed CM $C$.
With this in mind, the rest of this section focuses on how to simulate well-formed counter machines using CBM cells.

\paragraph{Simulating Counter Machines with CBM Cells.}
For a given well-formed CM $C$,
our simulation strategy leverages the following cell types:
{\em counter}, {\em state}, and {\em transition} cells. The simulation requires one counter cell for each counter in the simulated
machine, one state cell for each state in the machine's finite state control, and one transition cell for each transition.
As detailed below, our cell definitions work for a minimum binding bound of $1$, 
and no cell definition includes more than $2$ bioelectric events or requires its membrane function to react to more than $2$ different ligand types.
All firing functions used in these cells are simple deterministic step functions.

The basic dynamic of the simulation is that the counter cells each announce with a designated ligand whether their current value is $0$ or greater than $0$.
The state cell corresponding to the current state of the simulated machine also announces that it is the current state with a designated ligand (while other state
cells remain quiet).
Transitions for a well-formed CM are defined with respect to the current state and the count status of a single specified counter.
For a given transition cell defined with respect to state $q$, counter $c_i$, and status $s_i\in \{0, >0\}$ of $c_i$, 
we define its membrane function so that if it receives an active ligand from the state $q$ cell,
and the ligand corresponding to status $s_i$ from the counter cell corresponding to $c_i$,
then its potential increases to a point that triggers two events: one that transmits a ligand that 
actives the state cell corresponding to the new state of the simulated machine,
and one that transmits a ligand that implements a counter increment or decrement 
when received by the corresponding counter cell.
At this point, having activated a new state and adjusted a counter value,
the transition cell falls back to its initial potential and the system is ready to simulate a new step. 



We detail this general strategy below, starting with our definition of well-formed counter machines, and then moving on
to the detailed definitions of each of the summarized cell types.

\paragraph{Well-Formed Counter Machines.}
Give some counter machine $C'$, we can transform it, without loss of generality, into 
an equivalent (in terms of its final outputs) {\em well-formed} machine $C$ that satisfies the following properties:

\begin{enumerate}
 \item The transitions defined for each state $q$,
          are defined with respect to the $\{0, >0\}$ status of exactly one counter,
          or a wildcard $*$ that specifies the transition should occur regardless of any particular counter statuses.
 \item Each transition performs a single operation (increment or decrement) on a single counter in addition to transitioning
 the machine to a new state. 
  \item The transitions are complete and unambiguous: regardless of the counter statuses and current state,
 exactly one transition is enabled at any given step.
 \item The machine never decrements a counter with value $0$.

 \end{enumerate}

\paragraph{The Counter Cell Type.}
We define one counter cell $C_i$ for each counter $c_i$ from the well-formed counter machine we are simulating.
The potential value of a cell represents the  value of the counter it is simulating,
accordingly, we initialize $C_i$'s potential to the appropriate initial count value for $c_i$.
We also set
the cell's gradient rate $\lambda \gets 0$, meaning that it
has no background drift impacting its potential value (as would be the case, for example, if there were no ion channels
enabling ion flux between the cell and the extracellular environment).
The cell has
a bioelectric event that sends a $ZERO_i$ ligand with probability $1$ if its potential is less than or equal to $0$,
and a bioelectric event that sends a $NONZERO_i$ ligand with probability $1$ if its potential is greater than $0$.

The membrane function for $C_i$ considers only the counts of two special ligands: $INC_i$ and $DEC_i$.
If at least one $INC_i$ ligand is received, it increments the potential by $1$. If at least one $DEC_i$ ligand is received,
it decrements the potential by $1$. (As will be made clear, $C_i$ can never receive both $INC$ and $DEC$
ligands in the same round, so any behavior for this case is fine).

\paragraph{The State Cell Type.}
We define one state cell $Q_i$ for each state $q_i$. 
Each cell has a gradient rate $\lambda = 0$.
The cell corresponding to the start state is initialized with potential $1$,
all other state cells are initialized with potential $0$.
The cell's definition includes a single bioelectric event that sends a $STATE_i$ ligand with probability $1$ if the cell's potential is at least $1$.
If this event fires, it also decreases the cell's potential by $1$.
The membrane function for the cell considers only the $ASTATE_i$ ligand.
If it receives any ligands of this type, it increases the cell's potential by $1$.

\paragraph{The Transition Cell Type.}
A well-formed counter machine transition $t_i$ is defined for a counter state $q_j$, 
a counter $c_k$, and a status for this counter $s_k\in \{0, >0\}$.
This transition is enabled if the machine matches this $(q_j, c_k, s_k)$ transition precondition,
meaning that the machine is currently in state $q_j$ and the status of $c_k$ is $s_k$.\footnote{As defined earlier, a
 well-formed counter machine can also include a transition of the form $(q_j, *)$ meaning
 that the transition is enabled if the the machine is in state $q_j$ regardless of any counter status.
 We can replicate this case with two transition preconditions, $(q_j, c_1, 0)$ and $(q_j, c_1, >0)$
 that map to the same behavior as the original $(q_j, *)$ transition. This has the same effect
 as exactly one of those two preconditions must be enabled if the machine is in state $q_j$, 
 as counter $c_1$---as with any counter---must either have status $0$ or $>0$, but not both.}
Recall from the definition of well formed that at most one transition can be enabled at any given step.

Fix some such transition for $t_i$ with precondition $(q_j, c_k, s_k)$.
Let $(q_{j'}, c_{k'}, a_{k'})$, for state $q_{j'}$, counter $c_{k'}$, and action $a_{k'} \in \{INC_{k'}, DEC_{k'}\}$,
be the transition result. That is, if enabled, this transition shifts the machine to state $q_{j'}$ and executes action
$a_{k'}$ on counter $c_{k'}$. 
We define a transition cell $T_i$ for $t_i$ as follows.

The cell has initial potential $0$, gradient rate $\lambda = 1$, and equilibrium $\sigma = 0$.
If $s_k = 0$, then its membrane function increases the cell's potential by $1$ if it receives a $STATE_j$ and $ZERO_k$ ligand.
If $s_k  >0$, then its membrane function increases the cell's potential by $1$ if it receives a $STATE_J$ and $NONZERO_k$ ligand.
These are the only multisets that the function maps to a potential change.

The cell definition includes two bioelectric events.
The first sends a $ASTATE_{j'}$ ligand with probability $1$ if the cell's potential is at least $1$. 
The second depends on $a_{k'}$.
If $a_{k'} = INC_{k'}$ m then the second event sends $INC_{k'}$ ligand with probability $1$ if the cell's potential is at least $1$,
otherwise it sends $DEC_{k'}$.
Neither event adjusts the cell's potential.
But because $\lambda = 1$ and $\sigma = 0$,
if the cell starts the round at potential $1$, it will drift back to $0$ at the end of the round as it drives its potential toward equilibrium.

\paragraph{Pulling Together the Pieces.}
To understand the operation of this simulation strategy, it is helpful to divide the rounds into alternating $a$ and $b$ types,
starting with type $a$. Each consecutive pair of $a$ and $b$ rounds simulates one round of the counter machine.

In more detail, at the start of the first $a$ round, all transition cells are resting at potential $0$.
The only state cell at potential $1$ is the cell corresponding to the initial state $q_0$.
This cell will send a $STATE_0$ ligand and reduce its potential back to $0$.
Each counter cell $c_i$ will also send a $ZERO_i$ or $NONZERO_i$ ligand, depending on its potential value.
The transition cell $T_j$ corresponding to the unique enabled transition for this round will
be the only transition cell to receive the proper combination of state and counter ligands to increase
its potential by $1$. 

At the beginning of the $b$ round that follows, the only cells that will send ligands are the counter cells
and the transition cell $T_j$.
The ligands from the counter cells will be effectively ignored by the membrane functions of the transition cells in this round as they do not
come accompanied by any state ligands.
The transition cell $T_j$, by contrast,
will send a ligand $ASTATE_k$, where $q_k$ is the new state in the result for transition $t_j$,
and $INC_{\ell}$ or $DEC_{\ell}$ ligand, corresponding to the counter operation on some counter $c_{\ell}$ specified in the result for $t_j$.
Because $\sigma=0$ and $\lambda=1$, $T_j$ will conclude the round by having its potential reduce back
down to $0$.
The $INC_{\ell}$ or $DEC_{\ell}$ ligands will have the effect of properly updating the potential value of counter $c_{\ell}$.
The $ASTATE_k$  ligand will increase the potential of the state cell corresponding to $q_k$ to $1$.

When we start the subsequent $a$ round, we are back in a configuration where all transition
cells have potential $0$, and the only state cell with a potential $1$ is the state cell corresponding
to the simulated machine's new state. 
The simulation of this next round continues as above.

